\documentclass[11pt]{article}
\usepackage{authblk}

\usepackage{amsfonts}
\usepackage{amsmath}
\usepackage{amssymb}
\usepackage{amsthm}
\usepackage{array}
\usepackage{indentfirst}
\usepackage[left=0.95in,top=0.9in,bottom=0.9in,right=0.95in]{geometry}
\usepackage[round]{natbib}
\usepackage{graphicx}
\usepackage{enumerate}
\usepackage[export]{adjustbox}
\usepackage[font=small,labelfont=bf]{caption}
\usepackage{float}
\usepackage[usenames,dvipsnames,svgnames]{xcolor}
\usepackage{tikz}
\usetikzlibrary{calc}
\usepackage{graphics,color}
\usepackage{mathpazo}
\usepackage{bbold}
\usepackage{xr-hyper}

\usepackage{etoolbox}
\usepackage{versions}
\usepackage{placeins} %
\usepackage[nodisplayskipstretch]{setspace}
\usepackage[%
plainpages,%
colorlinks=true,%
urlcolor=blue,%
filecolor=black,%
citecolor=black,%
pdfpagemode=UseOutlines,%
pdfauthor={Ian M. Schmutte and Nathan Yoder},%
pdfsubject={Information Design for Differential Privacy},%
]{hyperref}

\newtheoremstyle{pfof}
{\topsep}
{\topsep}
{\normalfont}
{0pt}
{\bfseries}
{}
{5pt plus 1pt minus 1pt}
{}
\theoremstyle{pfof}
\newtheorem*{proofof}{Proof of}
\theoremstyle{definition}
\newtheorem*{defn}{Definition}

\newtheorem{example}{Example}
\theoremstyle{plain}
\newtheorem{thm}{Theorem}
\newtheorem{cor}{Corollary}
\newtheorem{lem}{Lemma}

\newtheorem{prop}{Proposition}
\theoremstyle{remark}

\DeclareMathOperator{\Gr}{Gr}

\DeclareMathOperator{\spn}{span}
\DeclareMathOperator{\supp}{supp}

\DeclareMathOperator{\conv}{conv}

\DeclareMathOperator{\ri}{ri}
\DeclareMathOperator{\ext}{ext}
\DeclareMathOperator{\rbd}{rbd}
\DeclareMathOperator{\aff}{aff}

\newcommand{\hitype}{T}
\newcommand{\vf}{\hat{v}}
\newcommand{\ccvf}{\hat{V}}
\newcommand{\tvf}{v}
\newcommand{\tccvf}{V}
\newcommand{\util}{u}

\newcommand{\exper}{\tau}
\newcommand{\exspace}{\Delta(\Delta(\Omega))}
\newcommand{\mech}{\pi}
\newcommand{\texper}{\nu}
\newcommand{\tspace}{\Theta}
\newcommand{\texspace}{\Delta(\Delta(\tspace))}
\newcommand{\belief}{\mu}
\newcommand{\prior}{\belief_0}
\newcommand{\tbelief}{\beta}
\newcommand{\tprior}{\tbelief_0}
\newcommand{\constraint}{K_\Omega}
\newcommand{\tconstraint}{K}

\newcommand{\gmech}{\sigma_\epsilon^g}
\newcommand{\gpost}{\tau_\epsilon^g}
\newcommand{\bpost}{\tau_b}
\newcommand{\cdf}{R}

\newcommand{\quant}{Q}
\newcommand{\proj}{P}

\makeatletter
\newcommand*{\addFileDependency}[1]{%
  \typeout{(#1)}
  \@addtofilelist{#1}
  \IfFileExists{#1}{}{\typeout{No file #1.}}
}
\makeatother

\newcommand*{\myexternaldocument}[1]{%
    \externaldocument[][nocite]{#1}%
    \addFileDependency{#1.tex}%
    \addFileDependency{#1.aux}%
}
\myexternaldocument{iddp_supplement}

\newtoggle{final}
\newtoggle{blind}
\togglefalse{final}
\togglefalse{blind}

\usepackage[font=small,labelfont=bf,skip=5pt,position=bottom]{caption}
\usepackage{titlesec}
\titlespacing{\subsection}{0pt}{*2}{*1.5}
\titlespacing{\subsubsection}{0pt}{*2}{*1}
\titlespacing{\paragraph}{0pt}{*2}{*1}
\usepackage{dcolumn}
\newcolumntype{d}[1]{D{.}{.}{#1}}
\usepackage{siunitx}
\iftoggle{final}{\excludeversion{instructions}}%
{\includeversion{instructions}}

{\\\hline\\\end{tabular}}

	\newcommand{\mytitle}{Information Design for Differential Privacy}
	
	\newcommand{\myversion}{\today}

	\date{\myversion}
    
     \title{\mytitle
    }
    \author{Ian M. Schmutte and Nathan Yoder\thanks{Schmutte: University of Georgia, Terry College of Business, John Munro Godfrey Sr. Department of Economics; E-mail: \texttt{schmutte@uga.edu}. Yoder: University of Georgia, Terry College of Business, John Munro Godfrey Sr. Department of Economics; E-mail: \texttt{nathan.yoder@uga.edu}. 
    An extended abstract of this paper appeared in \textit{EC '22: Proceedings of the 23rd ACM Conference on Economics and Computation}.
    The authors wish to thank Gary Benedetto, Mark Fleischer, Kevin He, R. Vijay Krishna, Meg Meyer, Marco Ottaviani, John Quah, Marzena Rostek, Alex Smolin, Peter Norman S\o rensen, Nikhil Vellodi, Juba Ziani, seminar participants at Bank of Canada, Concordia University, Emory, Florida State, Paris School of Economics, Penn State, and Yale, and conference attendees at NASMES 2022, EC'22, the 2022 ASU Theory Conference, SAET 2023, 
    NBER Data Privacy 2024, and ESSET 2024
    for their helpful comments, as well as Cole Wittbrodt for excellent research assistance. 
    Claude, GPT, and Github Copilot were used in the production of this manuscript. Wolfram Mathematica was used in the production of the figures. Both of the authors acknowledge summer support from the University of Georgia through a Terry-Sanford Research Award. Schmutte is grateful for financial support from the Bonbright Center for the Study of Regulation. 
    }}

\begin{document}

\begin{titlepage}
       \maketitle
       \begin{abstract}

        Firms and statistical agencies must protect the privacy of the individuals whose data they collect, analyze, and publish.
           These organizations often do so by using publication mechanisms that satisfy \textit{differential privacy}.  We consider the problem of choosing such a mechanism to maximize the value of its output to end users. 
           We show that mechanisms which add conditionally independent noise to the statistic of interest --- like most of those used in practice --- are
           never without loss of generality when the statistic is a sum or average of magnitude data (e.g., income). But conversely, adding conditionally independent noise is \textit{always} optimal when the statistic is a count of data entries with a certain characteristic, and the underlying database is drawn from a symmetric distribution (e.g., if individuals' data are i.i.d.).
           When, in addition, data users view higher actions and higher values of the statistic as complementary (e.g., their payoffs are supermodular), we show that the simple \textit{geometric mechanism} is always optimal by using a novel comparative static that ranks information structures according to their usefulness in monotone decision problems.
           
           
       \end{abstract}
       
     \thispagestyle{empty}
       
\end{titlepage}

\clearpage

\section{Introduction}

To help people make better decisions, statistical agencies collect and then disseminate information about data.  
For instance, 
{information about}
the prevalence of disease can affect decisions that impact public health, while 
{information about}
the level of unemployment can affect the decisions of participants in the labor market.
But disseminating information about disease prevalence or unemployment also necessarily reveals information about whether any one person is sick or unemployed \citep{Dinur:2003:RIW:773153.773173}.

The desire to avoid such disclosures has prompted
technology firms and statistical agencies
to adopt mechanisms for publishing data that provide formal, mathematically provable guarantees of privacy.
Frequently, these guarantees take the form of \textit{differential privacy} \citep{dwork2014algorithmic,fioretto2024differentialprivacyoverviewfundamental}.
This privacy criterion --- which we describe in detail --- requires a data provider to randomize its output so as to provide an explicit quantitative bound on 
the amount that a change in any individual entry in the data can change the probability of producing any given output.

Many different \textit{publication mechanisms} --- or (stochastic) maps from true data to published output --- have been proposed for publishing statistics in a way that satisfies differential privacy (e.g., \cite{GhoshRoughgardenSunararajan:Universally:SIAM:2012}; \cite{geng2015optimal}).
Because differential privacy limits the amount of information that can be revealed about individuals, these generally aim to disclose information about aggregate quantities --- i.e., population statistics --- such as the proportion of respondents in the data that are unemployed, or the average of those respondents' incomes.
Comparisons between publication mechanisms have focused on their expected \textit{accuracy} --- the expected value of (a function of) the distance between the published output and the true value of the statistic.

However, there is much less clarity about how to maximize the  \textit{value to data users} of published statistics while preserving privacy.
Our paper sheds light on this question by introducing an information design approach to differential privacy. 
In particular, we endow data users with payoffs that depend on their actions and the true value of a population statistic, and explicitly model the way that their decisions about the former depend on beliefs about the latter. We assume that data users are Bayesian, and update these beliefs after observing the output of a signal about the contents of the database from which the statistic is drawn --- i.e., a publication mechanism.
Then, we consider the problem of a data provider who chooses that signal under commitment --- i.e., before accessing the data --- so as to maximize the data users' welfare, subject to a differential privacy constraint.

This is essentially the problem
faced by many statistical agencies and technology firms in the real world. The U.S.\ Census Bureau has used formal privacy requirements, including differential privacy, as its primary approach to disclosure limitation in several data products.\footnote{See \cite{abowd2023confidentiality} (2020 Decennial Census, in which the Bureau used \textit{zero-concentrated differential privacy}) and \cite{foote2019releasing,PSEOdocumentation} (Post-Secondary Employment Outcomes).
}
Technology firms such as Google, LinkedIn, Uber, Apple, and Meta have deployed differential privacy in applications ranging from the analysis and publication of aggregate user behavior to the training of AI models.\footnote{See \cite{Guevara2019,kurakin2024harnessinglargelanguagemodelsgenerate,mckenna2025scalinglawsdifferentiallyprivate} (Google), \cite{rogers2020linkedins} (LinkedIn), \cite{Near2018} (Uber), \cite{Apple} (Apple), and \cite{FacebookResearch2020} (Meta).}
Crucially, offering a meaningful differential privacy guarantee requires data providers like these to commit to a mechanism without reference to the underlying data.%
\footnote{%
For example, the U.S.\ Census Bureau designed the publication mechanism for the 2020 Redistricting Files using demonstration data based on prior census data; \cite{timeline} gives a timeline summarizing this process.
See Section \ref{S:Model} for a detailed discussion of 
why using the data to develop or calibrate the mechanism can make the mechanism's privacy guarantee meaningless.
}

The data provider in our model faces a special case of a more general problem: choosing a \textit{signal} under \textit{commitment} to influence the \textit{decision} of a Bayesian agent who observes it. This general problem is the focus of the extensive literature on information design (e.g., \cite{KamenicaBayesianPersuasion2011,taneva2019information,kamenica2019bayesian}). 
By introducing key tools from this literature to the study of differential privacy --- and developing new tools that apply to information design more broadly --- we provide new insights that have direct, practical applications. 
To make these insights concrete, we describe them in the context of an example.

\begin{example}[School Planning]\label{Ex:School}
    A school district needs to choose how many slots $a$ to create in 
    its first year classes
    for the next three years. It would like to choose a number that is as close as possible to the district's actual number $\omega$ of children under 4, who will enter school during that time.
For simplicity --- and for a more direct comparison with the existing literature on accuracy in differential privacy --- we assume that its loss function is simply squared error: $(\omega-a)^2$.

    Since the children that will enroll are not currently in school, the district must rely on statistics published by 
    a government agency.
    Suppose that the statistical agency has chosen to publish information about $\omega$ in a way that is differentially private, with a ``privacy loss budget'' (the bound on the amount any single record can change the log probability of an output) of $\epsilon=1$. The records $\theta_i$ in its database each give the number of children under 4 in a household $i$, which varies between 0 and 2. Thus, a change in a data record can change $\omega$ by at most $\Delta=2$.
    
    In keeping with actual practice, we assume that the statistical agency chooses to satisfy differential privacy by \textit{adding conditionally independent noise} to the true value of $\omega$; i.e., publishing $\omega+\eta$, where $\eta$ is a random variable whose distribution is independent of the structure of the database, conditional on $\omega$.\footnote{For instance, the U.S. Census Bureau has used the geometric mechanism to publish industry-specific counts of employed graduates of postsecondary education programs \citep{PSEOdocumentation}. It is also used in the private sector: in the early COVID-19 pandemic, Facebook used its continuous analogue, the Laplace mechanism, to publish counts of individuals that did not leave a small area over the course of the day \citep{FacebookResearch2020}.}
    In particular, suppose that 
    it chooses to add noise that has a \textit{two-sided geometric} (or equivalently, discretized Laplace) distribution with parameter $\epsilon/\Delta=1/2$.
    Its choice to do so is
    motivated by the results of \cite{GhoshRoughgardenSunararajan:Universally:SIAM:2012}, who show that if 
    data users can remap the published output to $a\in\mathbb{R}$, and
    the population statistic $\omega$ is a count of records with some characteristic, 
    then among all ways to ensure differential privacy by adding conditionally independent noise to $\omega$, this one 
    maximizes expected accuracy (the expected value of any decreasing function of the error $|a-\omega|$).
    \cite{GhoshRoughgardenSunararajan:Universally:SIAM:2012} call this the \textit{geometric} publication mechanism. 

We evaluate the performance of this mechanism --- in terms of the value of its output to the data user --- in a geography with $N=100$ households with i.i.d. data records $\theta_i$ with 
$\{P(\theta_i=t)\}_{t=0}^2=\{0.89, 0.09, 0.02\}$
.\footnote{The distribution of the number of children under 4 in a household observed in the 2019 5-Year American Community Survey within Athens-Clarke County, Georgia.}
    In this context,
    when the school district chooses $a$ after seeing the output of the geometric mechanism, its expected squared error is 5.09.
    But despite the optimality result of \cite{GhoshRoughgardenSunararajan:Universally:SIAM:2012}, we can do much better. 
    In fact, by solving an information design problem (Proposition \ref{P:DesignProblemT}), we can produce a mechanism that lets the district achieve an expected loss
    of 3.78 --- about 26\% better than by using the off-the-shelf geometric mechanism.
    
The reason this improvement is possible is that, as we show in Theorem \ref{T:OblivMag} and Corollary \ref{C:OblivSuboptimal}, \textit{in settings like this one, adding conditionally independent noise can be a suboptimal way to satisfy the differential privacy constraint.}
    In this example, for instance, considering alternative distributions of conditionally independent noise can only reduce expected loss by 3\% relative to the geometric mechanism.\footnote{In settings with magnitude data, the geometric mechanism can be outperformed by a mechanism that adds conditionally independent noise that is not geometric; see, e.g., \cite{geng2015optimal}.  But in this example, this approach can only reduce expected squared error to 4.95.} Instead, we can do much better by allowing the distribution of published output to depend on the entire database, rather than just the true population statistic $\omega$.
    This fact is somewhat counterintuitive, since the data users' payoff depends only on the database through $\omega$. Indeed, we are not aware of any discussion of this point in the literature on differential privacy.
As we show, taking an information design approach to the problem --- and thus thinking of differential privacy as a constraint on the \textit{distributions of posterior beliefs} about $\omega$ that can be induced by the data publisher --- allows us to be precise about 
the sense in which adding conditionally independent noise is not without loss of generality.
Instead of an off-the-shelf mechanism (e.g., geometric), the optimal publication mechanism can then be computed either by using the concavification approach of \cite{KamenicaBayesianPersuasion2011}, or by solving a linear program (e.g., \cite{Bergemann2019}).

But, as we show, there are other settings in which adding conditionally independent noise \textit{is} without loss. Moreover, we show that the geometric mechanism is the optimal way to do so in a very broad class of settings.

To illustrate this, suppose that the district must instead determine how many school buses to purchase for its program. The optimal number of buses depends not on the number of children who will enroll, but on the \emph{number of households} with such children: the population statistic $\omega$ is now a \textit{count} of \textit{binary} data entries, rather than the \textit{sum} of \textit{magnitude} data.
As we show in Theorem \ref{T:OblivEquiv}, as long as it views households as \textit{anonymous} --- i.e., 
their observations are exchangeable random variables ---
this makes it optimal to ensure differential privacy by adding conditionally independent noise.

But how should it add such noise? Suppose that
the district 
seeks to minimize the number of households per bus route plus expenditure on buses:
when there are $\omega$ households with children under 4 and the district buys $a$ buses, its payoff is
$u(a,\omega)=-\frac{\omega}{a}-ca,\text{ for some }c>0.
$
This decision problem is not in the expected accuracy-maximization class considered in \cite{GhoshRoughgardenSunararajan:Universally:SIAM:2012}.
But even though the district does not 
want its action to \textit{match} the statistic $\omega$, a rightward shift in its belief about $\omega$ still causes its optimal action to increase; i.e., it faces a \textit{monotone} decision problem.
In Theorem \ref{T:Geometric}, we show that this is the key feature that makes the geometric mechanism optimal.
\end{example}

The kind of decision problems considered in Example \ref{Ex:School} --- those in which a data user's decision depends on their belief about a population statistic --- are common in practice: an individual takes more precautions when she believes a pandemic pathogen is more prevalent in her community; a restaurateur adopts a more upscale menu when he believes neighborhood incomes are higher; a municipal government offers employers larger tax incentives when it believes more of its constituents are unemployed.
Moreover, as in Example \ref{Ex:School} (and each of the applications listed above), data users frequently face problems that are \textit{monotone}; i.e., that feature complementarity between higher actions and higher beliefs about the population statistic.\footnote{By higher beliefs, we mean beliefs that are shifted to the right; i.e., higher in the monotone likelihood ratio sense \citep{milgrom1981good} or (if payoffs are supermodular) in the sense of first-order stochastic dominance.}

Our results offer guidance to a data provider who views these problems as representative use cases for the information it publishes. If the population statistic of interest is a sum of magnitude data, adding conditionally independent noise may be suboptimal (Theorem \ref{T:OblivMag} and Corollary \ref{C:OblivSuboptimal}); the data provider can often do better with a publication mechanism that depends not only on the \textit{average} of the data, but on the sample \textit{distribution} more generally. On the other hand, if it is a count of binary data from respondents that are ex ante anonymous, adding conditionally independent noise is without loss (Theorem \ref{T:OblivEquiv}) --- and in the case of a monotone decision problem, the geometric mechanism is always optimal (Theorem \ref{T:Geometric}).

We also provide a general comparative static on the comparison of information structures that we use to establish Theorem \ref{T:Geometric}.
Because of the nature of the data provider's
problem, Theorem \ref{T:Geometric} cannot directly appeal to existing tools for ranking experiments according to their usefulness to a class of decision makers.\footnote{In particular, the information structures we need to compare are not Blackwell-comparable, and do not necessarily induce an MLRP-ordered set of posterior beliefs, so that comparative statics based on \cite{lehmann1988comparing} accuracy cannot be directly applied.}
Instead, we introduce \emph{Uniform-Peaked Relative Risk (UPRR) dominance}, which ranks
information structures higher when the posteriors they induce each concentrate relatively more mass around a \textit{peak} (e.g., their mode).
Theorem \ref{T:ProbDominance} shows that information structures higher in the UPRR order are more useful in \textit{any} monotone decision problem --- one where payoffs are supermodular, have the single crossing property, or, most generally, form an interval dominance order family in the sense of \cite{quah2009comparative}.
Since the geometric mechanism is UPRR-dominant over \textit{all} differentially private data publication mechanisms that depend only on the true value of the population statistic (Lemma \ref{L:Geometric_Dominates}), Theorem \ref{T:Geometric} follows. %

\paragraph{Related Literature}

Our paper complements a growing body of research addressing the challenge posed by \citet{Abowd2019} for economists to develop analytical tools that prioritize information quality in privacy-preserving data publication. 
Several recent papers treat the publication mechanism as fixed, and ask how the data provider should set the privacy budget, both as a theoretical matter \citep{AbowdSchmutte:Privacy:AER, Hsu:EconomicEpsilon:IEEE:2014,echenique2021screening} and as a practical matter \citep{ChettyFriedmanPracticalNBERw25626}. By contrast, we treat the privacy requirement as fixed, and try to find an optimal publication mechanism.  Other authors focus on how property rights in data should be assigned \citep{Jones2020, ArrietaIbarra2018}, or the design of privacy-preserving mechanisms \citep{pai2013privacy, eilat2021bayesian}.

Our paper also relates to prior work in computer science that examines the performance of differentially private publication mechanisms (e.g., \cite{geng2015optimal,koufogiannis2015optimality}). This literature generally focuses on the
expected accuracy of mechanisms that are \textit{oblivious}, i.e., add conditionally independent noise to the population statistic. 
Among these papers, two are closest to our work. First, as discussed in Example \ref{Ex:School}, \citet{GhoshRoughgardenSunararajan:Universally:SIAM:2012} show that when 
data users are interested in a count of binary data and can postprocess the mechanism's output based on their priors, the geometric mechanism maximizes expected accuracy among all oblivious mechanisms.
When, in addition, data users have priors about the population statistic but view the database as ambiguous, they show the geometric mechanism also outperforms mechanisms that 
are not oblivious.
In contrast, our Theorem \ref{T:Geometric} shows that with binary data, the geometric mechanism is the optimal differentially private oblivious mechanism \textit{whenever data users face monotone decision problems}, and that using an oblivious mechanism is always optimal with binary data and a symmetric prior (Theorem \ref{T:OblivEquiv}), but need not be with magnitude data (Theorem \ref{T:OblivMag}, Corollary \ref{C:OblivSuboptimal}).
Second, \cite{BrennerNissim:Impossibility:SIAM:2014} show that when the statistic of interest is a sum of magnitude data, there is no \textit{single} oblivious mechanism that maximizes accuracy
for \textit{every} nondecreasing loss function and \textit{every} prior, the way that \cite{GhoshRoughgardenSunararajan:Universally:SIAM:2012} show the geometric mechanism does in the binary data case. In contrast, Theorem \ref{T:OblivMag} and Corollary \ref{C:OblivSuboptimal}
compare the \textit{entire class} of oblivious mechanisms to those that depend on the database more generally, and show that the latter can strictly outperform the former with magnitude data.

This paper follows other recent work that 
uses the tools of information design to study privacy.
Many of these papers focus on questions about the optimal \textit{amount} of privacy protection, generally using concepts other than differential privacy: \cite{ichihashi2020online} considers the way that consumers will choose to flexibly disclose information to sellers; \cite{nageeb2022voluntary} and \cite{hidir2021privacy} consider the optimal regulation of such disclosure. We focus instead on a question similar to those considered by \cite{he2023private} and \cite{strack2024privacy}: how should information be published, given the presence of a privacy constraint? Our approach differs from those papers, however, because differential privacy only protects sensitive information \textit{imperfectly}. All of these (including our paper) are part of a larger body of research on the economics of privacy; see \cite{acquisti2016economics} for a survey.

Our characterization of the data provider's problem is also related to work on constrained information design. In \cite{le2019persuasion}, the sender's signal must pass through a channel with limited capacity, a constraint that holds in expectation over the induced posteriors; \cite{doval2018constrained} give Lagrangian methods for problems with many such constraints. Differential privacy instead requires each induced posterior to lie in a fixed feasible set, as in \cite{matyskova2019bayesian}, where the receiver's ability to gather his own costly information has the same effect. \cite{samuelsonsteiner2025robust} call constraints of this form \textit{separable} and, like us, exploit the fact that they admit a concavification approach.

Finally, our Theorem \ref{T:ProbDominance} contributes to the literature on the comparison of experiments following \cite{blackwell1953equivalent}, whose ordering ranks signals according to their usefulness for \textit{any} decision maker. 
\citeauthor{lehmann1988comparing}'s (\citeyear{lehmann1988comparing}) \textit{accuracy} ordering ranks information structures with the \textit{monotone likelihood ratio property} (MLRP) so that those which are more accurate are more useful to decision makers 
whose payoffs are single-crossing in actions and states \citep{persico1996information}. \cite{quah2009comparative} extend this result to the more general case where payoffs form an \textit{interval dominance order} family. \cite{kim2023comparing} shows that for information structures with the MLRP, a decrease in accuracy is equivalent to a garbling by antitone noise.
\cite{athey2018value}, on the other hand, consider more general orders on information structures defined by their usefulness in decision problems satisfying various monotonicity conditions, and obtain \cite{lehmann1988comparing}'s order as a special case. 
Theorem \ref{T:ProbDominance} complements these results, showing that UPRR dominance allows accuracy-based comparisons of information structures whose posteriors are not MLRP-ordered --- indeed, whose posteriors need not be ordered at all.

The paper proceeds as follows. 
Section \ref{S:Model} formally describes the setting we consider and the meaning of differential privacy within it. 
Section \ref{S:Publication_as_Info_Design} shows that the data provider faces a constrained information design problem 
and characterizes its solution.
Section \ref{S:Oblivious} asks when this problem's dimensionality can be reduced by 
considering only mechanisms that add conditionally independent noise (Theorems \ref{T:OblivMag} and \ref{T:OblivEquiv}).
Finally, Section \ref{S:SPM} gives our optimality result for monotone problems (Theorem \ref{T:Geometric}) and the comparative static on information structures that it relies on (Theorem \ref{T:ProbDominance}).
All proofs are given in Appendix \ref{S:Proofs}.

\section{Setting}\label{S:Model}

There is a state of the world, or \textit{database};
a data provider, or \textit{designer}, 
who chooses a mechanism for releasing information about the database; and a \textit{decision maker},
who uses that information to make better decisions.\footnote{We assume a single decision maker for simplicity: All of our conclusions are unchanged if 
the designer instead seeks to maximize the total welfare of multiple decision makers, 
so long as their decisions do not have externalities on other decision makers, i.e., each DM's payoff only depends on their own action and the population statistic. 
(In that case, 
the problem is the same as if a single agent chose the decision makers' actions to maximize the sum of their payoffs.)
}

\paragraph{Data}
A database is a list $\theta=(\theta_1,\ldots,\theta_N)$ of $N\geq 2$ \textit{observations} $\theta_i\in\{0,1,\ldots,\hitype\}$. Each observation represents the \textit{type} of an individual or household (a \textit{respondent}). We say that the database contains 
\textit{binary data} when $\hitype=1$ and
\textit{magnitude data} when $\hitype>1$.

The set of possible databases is thus $\Theta=\{0,1,\ldots,\hitype\}^N$. The decision maker and the designer have a common prior $\tprior\in\Delta(\tspace)$ over this set, with full support: $\supp(\tprior)=\tspace$.%
 \footnote{One motivation for the differential privacy criterion is that it does not depend on assumptions about the background knowledge of agents (``attackers'') that might wish to use the mechanism output to learn about a respondent's type for malicious purposes.
 We emphasize that the common prior assumption need not apply to such malefactors, who are not explicitly considered in our model.%
 }

\paragraph{Data Publication and Differential Privacy}

Before observing the database,
the designer chooses a \textit{data publication mechanism} $(S,\mech)$,
where $S$ is some countable set of outputs and $\mech:\tspace\to \Delta(S)$ 
maps each database to a distribution over outputs in $S$.%
\footnote{The assumption that $S$ is countable (or even finite) is without loss in our setting --- see Corollary \ref{C:SolutionT}\eqref{I:SolutionT_Independence}.}
The decision maker then observes the realization of $\mech(\cdot|\theta)$.
Without loss, $S$ only includes outputs that the mechanism can actually generate: for each $s\in S$, $\mech(s|\theta)>0$ for some $\theta\in\tspace$.

The designer's chosen mechanism must 
guarantee the privacy of each respondent's type.
In particular, 
for some \textit{privacy loss} $\epsilon>0$,
it must be \textit{$\epsilon$-differentially private}.

\begin{defn}[$\epsilon$-Differential Privacy \citep{dwork2006calibrating}]
A data publication mechanism $(S,\mech)$ is \textit{$\epsilon$-differentially private} ($\epsilon$-DP) if for all
outputs $s\in S$ and all
 databases $\theta,\theta'\in\tspace$ that are adjacent%
, in the sense that they differ in at most one entry
(i.e., such that for some $i\in\{1,\ldots,N\}$, $\theta_{-i}=\theta'_{-i}$),
we have
\begin{align}\left|\log\left(\frac{\mech(s|\theta')}{\mech(s|\theta)}\right)\right|\leq\epsilon.
\label{E:DP}
\end{align}
\end{defn}

Differential privacy limits the amount, in log terms, that changing a single respondent's type can change the distribution of the data publication mechanism's realizations.\footnote{Because there is no uncertainty about the number of respondents in our model, our definition of differential privacy is referred to as \emph{bounded} differential privacy \citep{nofreelunch2011}, as distinct from \emph{unbounded} differential privacy, which bounds the change to the signal distribution from adding or removing a respondent from the database.
Since mechanisms that satisfy the latter criterion must take as argument databases of multiple sizes, the problem of designing them is somewhat different than the problem we consider. 
However, the key observation behind Theorems \ref{T:OblivMag} and \ref{T:OblivEquiv} --- with magnitude data (but not binary data), satisfying differential privacy requires oblivious mechanisms to treat databases similarly whenever they have similar population statistics, not just whenever they are adjacent ---
is still true 
in the unbounded case, and so we conjecture that those results have counterparts for unbounded DP.
}
This criterion is prior-independent: Whether a mechanism is differentially private does not depend on $\tprior$.
Hence, even if other agents (``attackers'') have more information about the contents of the database than the designer does, a differential privacy guarantee is still valid.
In fact, if an agent views types as independent \textit{a priori}, differential privacy is equivalent to a limit on the amount that observing the mechanism's output can shift her beliefs about an individual respondent's type $\theta_n$, regardless of how much information she already has about the other respondents' types. (See Proposition S.1 in the Online Appendix.)\footnote{\cite{mcsherry2007mechanism} give a different equivalent formulation of differential privacy based on test functions $g:S\to\mathbb{R}_+$. Formally, they show that a publication mechanism is $\epsilon$-DP if for every
nonnegative function $g:S\to\mathbb{R}_+$, $E_{\mech(\cdot|\theta)}[g(s)]\leq e^\epsilon E_{\mech(\cdot|\theta')}[g(s)]$ for each pair of adjacent databases $\theta,\theta'\in\tspace$.}

\paragraph{The Decision Maker}

After observing the realization of $\mech(\cdot|\theta)$, the decision maker takes an action $a$ from some compact set $A$. 
His payoff from doing so depends on $\theta$ only through the \textit{average} type of the respondents. Hence, since $N$ is fixed, his Bernoulli payoff can be written as a function $\util:A\times\Omega\to\mathbb{R}$ of his action $a$ and the \textit{population statistic} $\omega_\theta\equiv \sum_{n=1}^{N}\theta_n$,\footnote{In the terminology of the formal privacy literature, the population statistic is the answer to a \textit{query}: a deterministic function of the database (here, a sum). 
The query is distinct from the publication mechanism, which must produce a randomized output in order to satisfy \eqref{E:DP}.
} 
where $\Omega\equiv\{0,\ldots,N\hitype\}$; we assume that $\util$ is continuous.
Thus, if the decision maker has posterior belief 
$\tbelief\in\Delta(\Theta)$,
his interim payoff is given by $\tvf(\tbelief)\equiv\max_{a\in A}E_\tbelief\util(a,\omega_\theta)$.

\paragraph{The Designer's Objective}

When choosing the data publication mechanism, the designer acts to maximize the 
expected value of the decision maker's payoffs.
Formally, letting $\tbelief_s^{\mech}\in\Delta(\Theta)$ denote the posterior belief of a decision maker who observes $s$ from the data publication mechanism $(S,\mech)$, the designer solves
\begin{align}
\max_{(S,\mech)}&\left\{\sum_{\theta\in\Theta}\sum_{s\in S}\tvf(\tbelief_s^{\mech})\mech(s|\theta)\tprior(\theta)\text{ s.t. }
\left|\log\left(\frac{\mech(s|\theta')}{\mech(s|\theta)}\right)\right|\leq\epsilon 
\begin{array}{c}
\forall s\in S,\forall \theta,\theta'
\text{ s.t.}\\
\exists n,\theta_{-n}=\theta'_{-n}
\end{array}
\right\}.\label{E:OriginalProblem}
\end{align}

\paragraph{Discussion}

Our model captures a widely confronted data publication problem in very general terms. 
The data provider (designer) chooses how the data should be published, and commits to those choices after the structure of the data is determined, but before the actual data are collected.
The data provider is also under an obligation to guarantee differential privacy at a fixed level.

With binary data, our model corresponds precisely to the publication of information about the number of people in some category --- e.g., the 
number of unemployed workers in a specific occupation. Such statistics are generally called population or frequency counts, or --- in the computer science literature --- counting queries.
With magnitude data, our model corresponds to the publication of information about the sum or average of some characteristic --- e.g., income, years of schooling, or length of unemployment. These statistics are generally referred to as population totals or averages, or in the computer science literature, sum queries.
This focus on a single count, total, or average restricts data users' payoffs in a substantive way: it rules out, for example, data users who care about both the mean and the variance of income.
However, it does not enter into the definition of differential privacy, which restricts the distribution of the mechanism's output conditional on adjacent values of the database, rather than nearby values of the population statistic.\footnote{In Section \ref{S:Oblivious}, we show that when the publication mechanism only depends on the database through $\omega_\theta$, the privacy criterion \eqref{E:DP} can be simplified to compare the distribution of output conditional on nearby values of the population statistic. However, as we show in Theorem \ref{T:OblivMag}, these mechanisms are not without loss of generality in many settings.}

The notion of differential privacy that we consider is standard. But unlike most of the literature, we explicitly consider the problem of \textit{choosing} a publication mechanism that satisfies it. 
This raises the possibility that privacy might be compromised by the designer's choice of mechanism, rather than by the mechanism's output. For instance, if the designer were to learn the contents of the database $\theta$ before choosing a mechanism, she could perfectly reveal $\theta$ to the decision maker by choosing a 0-differentially private mechanism that always outputs $\theta$: $\mech(\theta|\theta')=1$ for all $\theta'\in\Theta$.
Thus, for the differential privacy guarantee to be meaningful, the publication mechanism must be chosen \textit{before} the designer learns the contents of the database --- as we require in our model.

However, even if the designer does not observe the database before choosing a publication mechanism, she may still be better informed than an attacker about the data generating process from which the database is drawn. 
Her choice of mechanism may thus allow such an attacker to infer that information, and update his beliefs about an individual respondent's type in a way that 
the mechanism  --- and hence the differential privacy criterion \eqref{E:DP} --- does not account for. 
This is not a privacy violation in the sense considered by differential privacy, because the extra information learned \textit{does not come from the database being protected.}\footnote{As \cite{dwork2014algorithmic} put it, ``Differential privacy promises to protect individuals from any additional harm that they might face due to their data
being in the private database $x$ that they would not have faced had
their data not been part of $x$. Although individuals may indeed face
harm once the results $\mathcal{M}(x)$ of a differentially private mechanism $\mathcal{M}$
have been released, differential privacy promises that the probability of
harm was not significantly increased by their choice to participate.'' (\cite{dwork2014algorithmic}, p. 20)}

Like most of the formal privacy literature, we assume the true value of the population statistic enters directly into the payoffs of data users (decision makers). 
We thus abstract away from other sources of uncertainty --- like measurement error and sampling variability --- about the extent to which the undistorted data reflect some true underlying state of the world that data users actually care about. This is without loss of generality, 
if we interpret the data users' payoffs as their expected utility conditional on the true value of the population statistic.

Our information design setting is related to the mechanism design setting considered by \cite{pai2013privacy}. They apply differential privacy to a map from respondents' type profiles to outcomes that determine their payoffs (i.e., a social choice function) rather than a map from databases to payoff-irrelevant signals (a publication mechanism). 
As they show, this ensures that the SCF is approximately dominant strategy implementable for the respondents.
In contrast, we consider a decision maker whose payoffs depend on the underlying database, not the mechanism's output.
Instead of explicitly modeling respondents' incentives, we simply require the mechanism to provide them with a guarantee of $\epsilon$-DP.

\section{Data Publication as Information Design \label{S:Publication_as_Info_Design}}

The differential privacy condition \eqref{E:DP} can be reformulated as a restriction on the 
set of posterior beliefs that the mechanism can induce.
Observe that if
 $\tbelief\in\Delta(\tspace)$ is the posterior belief of an agent who views realization $s$ from the data publication mechanism $(S,\mech)$, Bayes' rule tells us that for any two databases $\theta$ and $\theta'$ that differ in a single entry,
\begin{align}\frac{\tbelief(\theta')}{\tbelief(\theta)}=
\left.\frac{\mech(s|\theta')\tprior(\theta')}{\sum_{t\in\tspace}\mech(s|t)\tprior(t)}\right/\frac{\mech(s|\theta)\tprior(\theta)}{\sum_{t\in\tspace}\mech(s|t)\tprior(t)}=
\frac{\mech(s|\theta')}{\mech(s|\theta)}\frac{\tprior(\theta')}{\tprior(\theta)}.\label{E:Bayes}
\end{align}
Hence, 
$(S,\mech)$ is $\epsilon$-differentially private if and only if each posterior belief $\tbelief$ that it induces satisfies
\begin{align}
    \left|\log\left(\frac{\tbelief(\theta)}{\tbelief(\theta')}\right)-\log\left(\frac{\tprior(\theta)}{\tprior(\theta')}\right)\right|\leq\epsilon
    \text{ for each }\theta,\theta'\in\tspace\text{ with }\theta_{-n}=\theta'_{-n}\text{ for some }n.
    \label{E:DP_posteriorT}
\end{align}

We call the set of posteriors that satisfy \eqref{E:DP_posteriorT} the \emph{$\epsilon$-differentially private ($\epsilon$-DP) posteriors} $\tconstraint(\epsilon,\tprior)$.
In words, the $\epsilon$-DP posteriors are those such that for any 
two databases $\theta$ and $\theta'$ that differ in a single entry,
the difference in 
the log ratio of their probabilities under the prior and posterior is at most $\epsilon$.
As we show in the appendix (Lemma \ref{L:tconstraint}), 
this set
is a closed convex polyhedron which does not intersect the edges of the probability simplex, and which contains the prior in its (relative) interior.
Figure \ref{F:DPpostT} illustrates.
\begin{figure}[h!]
{\centering
\begin{tikzpicture}[
    x={(0.75cm,-0.65cm)},
    y={(1.75cm,0.75cm)},
    z={(0.35cm,1.95cm)},
    font=\small
]
\coordinate (theta00) at (-1,-0.707107,0);
\coordinate (theta01) at (0,0.707107,1);
\coordinate (theta10) at (0,0.707107,-1);
\coordinate (theta11) at (1,-0.707107,0);

\coordinate (dpbottom) at (0,-0.487391,0);
\coordinate (dptop) at (0,0.107646,0);
\coordinate (dpleft) at (-0.575210,-0.361008,0);
\coordinate (dpright) at (0.575210,-0.361008,0);
\coordinate (dpfront) at (0,-0.091189,0.331689);
\coordinate (dpback) at (0,-0.091189,-0.331689);
\coordinate (databaseprior) at (0,-0.235702,0);

\draw[line width=0.7pt]
    (theta00)--(theta01)--(theta10)--(theta00)--(theta11)--(theta01)
    (theta11)--(theta10);
\foreach \a/\b/\c in {
    dpbottom/dpleft/dpback,
    dpbottom/dpback/dpright,
    dptop/dpleft/dpback,
    dptop/dpback/dpright,
    dpbottom/dpleft/dpfront,
    dpbottom/dpfront/dpright,
    dptop/dpleft/dpfront,
    dptop/dpfront/dpright}{
    \path[draw=MidnightBlue!80!black, draw opacity=0.75,
        fill=MidnightBlue, fill opacity=0.20, line width=0.35pt]
        (\a)--(\b)--(\c)--cycle;
}
\fill (databaseprior) circle (1.15pt);
\node[below=10pt,xshift=-7pt] at (databaseprior) {$\tprior$};
\node[left=2pt] at (theta00) {$\theta=(0,0)$};
\node[right=2pt] at (theta01) {$\theta=(0,1)$};
\node[right=2pt] at (theta10) {$\theta=(1,0)$};
\node[left=2pt] at (theta11) {$\theta=(1,1)$};
\end{tikzpicture}
\caption{\label{F:DPpostT}\textbf{$\epsilon$-DP posteriors.
} 
The shaded region of $\Delta(\tspace)$ is the set $\tconstraint(\epsilon,\tprior)$ of $\epsilon$-DP posteriors
bounded by the constraints in (\ref{E:DP_posteriorT}), shown here in a setting with binary data ($\hitype=1$),
$N=2$, $\epsilon=1$, and prior $(\tprior(0,0),\tprior(0,1),\tprior(1,0),\tprior(1,1))=(
\frac{1}{3},\frac{1}{6},
\frac{1}{6},\frac{1}{3})$.
}}
\end{figure}

By representing $\epsilon$-differential privacy as a restriction to mechanisms that induce $\epsilon$-DP posteriors, we can write the designer's problem as a constrained information design problem: choosing a Bayes-plausible distribution $\texper\in\texspace$ of posteriors (i.e., a distribution whose expected value is the prior) whose support is constrained to lie in the set $\tconstraint(\epsilon,\tprior)$.

\begin{prop}[Differentially Private Data Publication as Information Design]
\label{P:DesignProblemT}The mechanism $(S,\mech)$ solves the designer's problem (\ref{E:OriginalProblem}) if and only if it induces a distribution of posteriors 
which solves %
\begin{align}
\max_{\texper\in\Delta(\tconstraint(\epsilon,\tprior))}\{E_\texper \tvf(\tbelief)\text{ s.t. }E_\texper\tbelief=\tprior\}.\label{E:DesignProblemT}
\end{align}

\end{prop}

Since the designer's incentives are aligned with the decision maker's, his value function in (\ref{E:DesignProblemT}) is convex. If the designer were not constrained by differential privacy, 
her problem would thus have a straightforward solution:
simply induce the posteriors at the vertices of $\Delta(\tspace)$ by publishing the true value of $\theta$.
With an $\epsilon$-DP mechanism, this is impossible: she can only induce posteriors in $\tconstraint(\epsilon,\tprior)$. But convexity --- or equivalently, \citeauthor{blackwell1953equivalent}'s (\citeyear{blackwell1953equivalent}) theorem --- still allows her to restrict attention to certain posteriors; namely, the extreme points (i.e., vertices) of the polyhedron $\tconstraint(\epsilon,\tprior)$.

This sharpens the implications of the information design literature for the designer's problem. For any $K\subseteq\Delta(\tspace)$, denote the restriction of $\tvf$ to $K$ as  $\tvf_K:K\to\mathbb{R}\cup\{-\infty\}$. Following \cite{KamenicaBayesianPersuasion2011}, define the \textit{$K$-restricted concavification} $\tccvf_K:K\to\mathbb{R}$ of $\tvf_K$ 
as the smallest concave function that lies above $\tvf$ on $K$.\footnote{Formally, let $\tccvf_K(\tbelief)\equiv\sup\{z|(\tbelief,z)\in\conv(\Gr(\tvf_K))\}$, where  $\Gr(\tvf_K)\equiv\{(\tbelief,\tvf_K(\tbelief))|\tbelief\in K\}$
denotes the graph of $\tvf_K$.}

\begin{cor}[Characterization of Optimal Data Publication Mechanisms]\label{C:SolutionT}
\hspace{0pt}
\begin{enumerate}[i.]
\item The maximized value of the designer's problem (\ref{E:OriginalProblem}) is $\tccvf_{\tconstraint(\epsilon,\tprior)}(\tprior)$.
\label{I:SolutionT_Max}
\item 
There is a mechanism which solves the designer's problem (\ref{E:OriginalProblem}) and induces a distribution of posteriors $\texper^*\in\Delta(\Delta(\tspace))$ such that
\label{I:SolutionT_ArgMax}
\begin{enumerate}
\item Each $\tbelief\in\supp\texper^*$ is an extreme point of $\tconstraint(\epsilon,\tprior)$;\label{I:SolutionT_Extreme}
\item 
Each $\tbelief\in\supp\texper^*$ achieves the privacy bound \eqref{E:DP_posteriorT} between at least $(\hitype+1)^N-1$ pairs of databases:
 $\left|\log\left(\frac{\tbelief(\theta)}{\tbelief(\theta')}\right)-\log\left(\frac{\tprior(\theta)}{\tprior(\theta')}\right)\right|=\epsilon$ for at least $(\hitype+1)^N-1$ distinct pairs $\{\theta,\theta'\}$ that have $\theta_{-i}=\theta'_{-i}$ for some $i$;\label{I:SolutionT_BoundAttained}
\item 
The support of $\texper^*$ is a linearly independent set of vectors in $\mathbb{R}^{(\hitype+1)^N}$; and thus\footnote{For concreteness, let each belief $\tbelief\in\Delta(\tspace)$ be represented by the vector in $\mathbb{R}^{(\hitype+1)^N}$ whose $n$th entry corresponds to the probability $\tbelief(\theta)$ it places on the database $\theta$ represented by $n$ in base $(\hitype+1)$.}
\label{I:SolutionT_Independence}
\item $\texper^*$ is the unique Bayes-plausible distribution of posteriors with support $\supp\texper^*$.
\label{I:SolutionT_UniqueSupport}
\end{enumerate}
\end{enumerate}
\end{cor}

Corollary \ref{C:SolutionT} applies standard results from information design to give a characterization of optimal differentially private data publication that we use to establish our main results.
Part \eqref{I:SolutionT_Max} shows that the value of the problem has the familiar concavification characterization from \cite{KamenicaBayesianPersuasion2011} --- but 
concavification takes place on the set of $\epsilon$-DP posteriors, rather than the entire simplex.
This places additional structure on the problem's solution: 
Part \eqref{I:SolutionT_ArgMax} places an upper bound on the number of posteriors induced by the mechanism --- and hence the number of outputs it produces \eqref{I:SolutionT_Independence} --- and a lower bound on the number of privacy bounds that each of those posteriors must attain \eqref{I:SolutionT_BoundAttained}. It also allows us to focus on the \textit{set} of posteriors that the mechanism induces --- which, without loss, is linearly independent \eqref{I:SolutionT_Independence} --- and ignore the probabilities with which it induces them \eqref{I:SolutionT_UniqueSupport}. Note that while \eqref{I:SolutionT_ArgMax} characterizes a particular solution to \eqref{E:DesignProblemT}, multiple solutions may exist, since $\tvf$ is not strictly convex.

\section{Conditionally Independent Noise and Oblivious Mechanisms}\label{S:Oblivious}
Section \ref{S:Publication_as_Info_Design} shows that the designer's problem amounts to choosing a Bayes-plausible distribution of posterior beliefs about the database. 
This problem is challenging in part because of its dimensionality: 
The space of posteriors about the database $\Delta(\tspace)$ has dimension  $(\hitype+1)^N-1$, the number of possible databases minus one.

But because the decision maker's payoff only depends on the population statistic $\omega$, his interim payoff only depends on his belief $\tbelief\in\Delta(\tspace)$ about the database through its \textit{projection} onto the lower-dimensional space $\Delta(\Omega)$ of beliefs about the population statistic.
That is, letting $\proj:\Delta(\tspace)\to\Delta(\Omega)$ be the projection operator defined by $\proj\tbelief(\omega)=\sum_{\theta:\omega_\theta=\omega}\tbelief(\theta)$, we have
\begin{align}\tvf(\tbelief)\equiv \max_{a\in A}E_{\tbelief}u(a,\omega_\theta)&=\max_{a\in A}E_{\proj\tbelief}u(a,\omega)\equiv\vf(\proj\tbelief).\label{E:vfproj}
\end{align}
Since the designer works to maximize the expected value of this payoff, we might expect that she can restrict attention to mechanisms that 
depend only on
the true value of the population statistic; i.e., those that add noise that is, conditional on $\omega$, independent of the structure of the database $\theta$.
Following the formal privacy literature, we refer to such mechanisms as \textit{oblivious}.
If oblivious mechanisms are without loss,
she can solve her problem \eqref{E:DesignProblemT} by choosing
a distribution on the set of posteriors about the population statistic $\Delta(\Omega)$ directly, rather than choosing a distribution of posteriors about the database and projecting them onto $\Delta(\Omega)$.

It turns out that
whether oblivious mechanisms are without loss of generality
depends on the type of data that the designer has. 
With magnitude data, oblivious mechanisms are never without loss of generality (Theorem \ref{T:OblivMag}). In fact, for every prior and privacy loss budget, there is a simple decision problem for which every oblivious mechanism is suboptimal (Corollary \ref{C:OblivSuboptimal}). But with binary data, oblivious mechanisms are always optimal when respondents are anonymous (Theorem \ref{T:OblivEquiv}) --- e.g., if the data are i.i.d.

\subsection{Characterizing Differential Privacy for Oblivious Mechanisms}

Formally, we say a data publication mechanism $(S,\mech)$ is \textit{oblivious} if 
$\mech(\cdot|\theta)=\mech(\cdot|\theta')$ whenever 
$\omega_\theta =\omega_{\theta'}$.%
\footnote{
In the terminology of the computer science 
literature on formal privacy, the population statistic $\omega_\theta$ is the outcome of a \emph{sum query} (with magnitude data) or \textit{counting query} (with binary data) applied to $\theta$, and our definition specifies that a mechanism is oblivious {with respect to that query}.
}
Hence,
there is a function $\sigma:\Omega\to\Delta(S)$ such that $\mech(\cdot|\theta)=\sigma(\cdot|\omega_\theta)$ for each $\theta\in\tspace$; we abuse notation and write $(S,\sigma)$ to denote such a mechanism. 
The class of oblivious mechanisms 
includes most differentially private mechanisms used in practice to publish information about counts or averages.
In particular, the widely used \textit{Gaussian} and \textit{geometric} mechanisms are each oblivious, since they publish the sum of the true population statistic $\omega_\theta$ and a random variable.

Because oblivious mechanisms only provide information about the population statistic, they are completely characterized by the distribution $\exper\in\exspace$ of posterior beliefs that they induce about the population statistic, rather than the distribution of beliefs they induce about the database more generally.
Such a distribution can be induced by an oblivious mechanism precisely when its expectation $E_\exper\belief$ is equal to the prior \textit{about the population statistic} $\prior\equiv\proj\tprior$.

When a mechanism is oblivious, Proposition \ref{P:DPObliv} shows that the differential privacy criterion \eqref{E:DP} simplifies to a limit on the amount that \textit{moving to an adjacent population statistic} (i.e., from $\omega-t$ to $\omega$, or vice versa, for $1\leq t\leq\hitype$) can change the distribution of the mechanism's realizations.
Intuitively, if two values $\omega,\omega'$ of the population statistic differ by no more than $\hitype$, there is a pair of databases with those statistics which differ in only one entry and by exactly the difference $\omega-\omega'$.
Consequently, differential privacy bounds the ratio of the posterior probabilities of those states \eqref{E:DP_posterior}.

\begin{prop}[Differential Privacy for Oblivious Mechanisms]\label{P:DPObliv}
Suppose $(S,\sigma)$ is an oblivious data publication mechanism. Then the following are equivalent:
\begin{enumerate}[i.]
    \item $(S,\sigma)$ is $\epsilon$-differentially private.\label{I:DPObliv:DP}
    \item $
    \left|\log\left(\frac{\sigma(s|\omega)}{\sigma(s|\omega-t)}\right)\right|\leq\epsilon$
     for each $s\in S$, $\omega\in\{1,\ldots,N\hitype\}$, and $1\leq t\leq \min\{\hitype,\omega\}$.
    \label{I:DPObliv:DPObliv}
    \item For each posterior belief about the population statistic $\belief\in\Delta(\Omega)$ induced by $(S,\sigma)$,\label{I:DPObliv:posterior}
    \end{enumerate}
    \vspace*{-\baselineskip}
    \begin{align}\begin{array}{@{}r@{}}
    \left|\log\left(\frac{\belief(\omega)}{\belief(\omega-t)}\right)-\log\left(\frac{\prior(\omega)}{\prior(\omega-t)}\right)\right|\leq\epsilon
    \end{array}
    \text{ for each }\omega\in\{1,\ldots,N\hitype\}\text{ and }1\leq t\leq \min\{\hitype,\omega\}.\label{E:DP_posterior}
\end{align}
\end{prop}

We call the set of posterior beliefs about the population statistic that satisfy \eqref{E:DP_posterior} the \textit{oblivious $\epsilon$-DP posteriors} $\constraint(\epsilon,\prior)\subset\Delta(\Omega)$. Figure \ref{F:DPpost} illustrates in the binary data setting described in Figure \ref{F:DPpostT}.

\begin{figure}[h!]
{\centering
\begin{tikzpicture}[font=\small]
\def\DPside{4.5}
\pgfmathsetmacro{\DPheight}{sqrt(3)*\DPside/2}
\pgfmathsetmacro{\DPexp}{exp(1)}
\pgfmathsetmacro{\DPedgehigh}{\DPexp/(1+\DPexp)}
\pgfmathsetmacro{\DPedgelow}{1/(1+\DPexp)}
\coordinate (statezero) at (0,0);
\coordinate (stateone) at ({\DPside/2},\DPheight);
\coordinate (statetwo) at (\DPside,0);
\coordinate (bluehigh) at ($(statezero)!\DPedgehigh!(stateone)$);
\coordinate (bluelow) at ($(statezero)!\DPedgelow!(stateone)$);
\coordinate (redhigh) at ($(stateone)!\DPedgehigh!(statetwo)$);
\coordinate (redlow) at ($(stateone)!\DPedgelow!(statetwo)$);
\coordinate (mux) at
    ({\DPside*(1+0.5*\DPexp)/(1+\DPexp+\DPexp*\DPexp)},
     {\DPheight*\DPexp/(1+\DPexp+\DPexp*\DPexp)});
\coordinate (muy) at
    ({\DPside/2},{\DPheight*\DPexp/(2+\DPexp)});
\coordinate (muz) at
    ({\DPside*(\DPexp*\DPexp+0.5*\DPexp)/(1+\DPexp+\DPexp*\DPexp)},
     {\DPheight*\DPexp/(1+\DPexp+\DPexp*\DPexp)});
\coordinate (muw) at
    ({\DPside/2},{\DPheight/(1+2*\DPexp)});
\coordinate (statprior) at ({\DPside/2},{\DPheight/3});

\draw[Blue, densely dashed, line width=0.45pt] (bluehigh)--(statetwo);
\draw[Blue, densely dashed, line width=0.45pt] (bluelow)--(statetwo);
\draw[Red, densely dashed, line width=0.45pt] (statezero)--(redhigh);
\draw[Red, densely dashed, line width=0.45pt] (statezero)--(redlow);
\path[fill=MidnightBlue!16, draw=black, line width=0.7pt]
    (mux)--(muy)--(muz)--(muw)--cycle;
\draw[line width=0.7pt] (statezero)--(stateone)--(statetwo)--cycle;
\fill (statprior) circle (1.15pt);
\node[below=1pt] at (statprior) {$\prior$};
\node[below left=1pt] at (statezero) {$\omega=0$};
\node[above=1pt] at (stateone) {$\omega=1$};
\node[below right=1pt] at (statetwo) {$\omega=2$};
\end{tikzpicture}
\caption{\label{F:DPpost}\textbf{Oblivious $\epsilon$-DP posteriors.
}
The  shaded region of $\Delta(\Omega)$ is the set of oblivious $\epsilon$-DP posteriors $\constraint(\epsilon,\prior)$ in the setting of Figure \ref{F:DPpostT}, where the prior about the population statistic is $\prior=\proj\tprior=
(\frac{1}{3},\frac{1}{3},\frac{1}{3})
$ and the privacy loss budget is $\epsilon=1$. The red dotted lines represent the upper and lower bounds on $\belief(1)/\belief(0)$ described in \eqref{E:DP_posterior}, while the blue dotted lines represent the bounds on $\belief(2)/\belief(1)$.
$\constraint(\epsilon,\prior)$ is asymmetric because there is no constraint on the ratio $\belief(0)/\belief(2)$, since population statistics 0 and 2 are not adjacent.
}}
\end{figure}

\subsection{When Are Oblivious Mechanisms Without Loss?}

Proposition \ref{P:DPObliv} reveals the key difference between differential privacy's restrictions on the information disclosed by an oblivious mechanism and its restrictions on the information disclosed \textit{about the population statistic} by a non-oblivious mechanism. 
In general, the designer can induce a distribution $\exper$ of posterior beliefs about the population statistic whenever it is the projection onto $\Delta(\Omega)$ of a Bayes-plausible distribution on the set of $\epsilon$-DP posteriors; i.e., whenever there exists a Bayes-plausible $\texper\in\Delta(\tconstraint(\epsilon,\tprior))$ such that %
$\exper(\belief)=\texper(\proj^{-1}(\belief))$ for each $\belief$.
The set of distributions that satisfy this criterion coincides with 
the set of distributions that can be induced with an $\epsilon$-DP  \textit{oblivious} mechanism
precisely when
\[\constraint(\epsilon,\prior)=\proj\tconstraint(\epsilon,\tprior),\]
i.e., when the projection of any $\epsilon$-DP posterior onto $\Delta(\Omega)$ is an oblivious $\epsilon$-DP posterior 
(Lemma \ref{L:OblivProj} in the appendix).

Whether this is the case depends crucially on the kind of data (and hence, the kind of population statistic) that the designer wishes to publish information about. With magnitude data ($T>1$), the answer is always no, regardless of the prior $\tprior$ or the privacy loss budget $\epsilon$: oblivious mechanisms are \textit{never} without loss of generality.

\begin{thm}[Oblivious Mechanisms and Magnitude Data]\label{T:OblivMag}
With magnitude data, $\constraint(\epsilon,\prior)\neq \proj\tconstraint(\epsilon,\tprior)$ for any database prior $\tprior$ and positive privacy loss budget $\epsilon$. Consequently, there exist $\epsilon$-DP data publication mechanisms that induce distributions of beliefs about the population statistic that cannot be replicated by an oblivious mechanism that is also $\epsilon$-DP.
\end{thm}

Theorem \ref{T:OblivMag} does not rule out the possibility that an oblivious mechanism might be optimal in a setting with magnitude data. Even though some distributions of beliefs about the population statistic cannot be induced by an oblivious mechanism, those distributions might all be suboptimal in \eqref{E:DesignProblemT}. However, as Corollary \ref{C:OblivSuboptimal} shows, we can always construct a simple decision problem such that only a non-oblivious mechanism can shift beliefs enough to change the optimal action.

\begin{cor}[Suboptimality of Oblivious Mechanisms]\label{C:OblivSuboptimal}
With magnitude data, for any database prior $\tprior$ and privacy loss budget $\epsilon$, there is a set of two actions $A=\{0,1\}$ and a
payoff function $u:\{0,1\}\times\Omega\to\mathbb{R}$
such that no $\epsilon$-DP data publication mechanism that solves the designer's problem \eqref{E:OriginalProblem} is oblivious.
\end{cor}

The intuition for Theorem \ref{T:OblivMag} is as follows.
Because an oblivious mechanism only depends on the database through the population statistic, 
viewing its output must change the log ratio of the probabilities of a pair of databases by the same amount whenever they have the same population statistics:
\begin{align*}
    \left|\log\left(\frac{\tbelief(\theta)}{\tbelief(\theta')}\right)-\log\left(\frac{\tprior(\theta)}{\tprior(\theta')}\right)\right|&=\left|\log\left(\frac{\tbelief(\hat{\theta})}{\tbelief(\hat{\theta}')}\right)-\log\left(\frac{\tprior(\hat{\theta})}{\tprior(\hat{\theta}')}\right)\right|
    \text{ whenever }\omega_\theta=\omega_{\hat{\theta}}\text{ and }\omega_{\theta'}=\omega_{\hat{\theta}'}.
\end{align*}
Requiring a mechanism to be oblivious thus adds constraints to \eqref{E:DP_posteriorT}
for pairs of databases $\theta,\theta'$ that do not differ in one entry, but have the same population statistics as some other pair $\hat{\theta},\hat{\theta}'$ that does. The key to establishing Theorem \ref{T:OblivMag} is showing that with magnitude data, these additional constraints never become redundant when they are projected onto the space $\Delta(\Omega)$ of beliefs about the population statistic. 
As we show in the Online Appendix, this gap is not an artifact of modeling the set of respondent types as finite, i.e., it persists as $\hitype\to\infty$.

But with binary data ($T=1$), projection onto $\Delta(\Omega)$ \textit{does} render these additional constraints redundant, so long as the prior distribution $\tprior$ is symmetric: $\tprior(\theta)=\tprior(\theta')$ whenever $\theta$ is a permutation of $\theta'$.\footnote{In other words, when respondent types $\{\theta_n\}_{n=1}^N$ are exchangeable random variables.}
When the prior is symmetric in this way, we say that \textit{respondents are anonymous}. Theorem \ref{T:OblivEquiv} shows that anonymity ensures that $\constraint(\epsilon,\prior)=\proj\tconstraint(\epsilon,\tprior)$ --- and hence that oblivious mechanisms are without loss.

\begin{thm}[Oblivious Mechanisms and Binary Data]\label{T:OblivEquiv}
With binary data and anonymous respondents, 
$\constraint(\epsilon,\prior)=\proj\tconstraint(\epsilon,\tprior)$:
For each $\epsilon$-DP data publication mechanism,
there exists an $\epsilon$-DP oblivious data publication mechanism 
that induces the same distribution of posterior beliefs 
about the population statistic.
\end{thm}

As a consequence of Theorem \ref{T:OblivEquiv}, whenever respondents are anonymous and data is binary, the designer's information design problem \eqref{E:DesignProblemT} simplifies to choosing a distribution on the lower-dimensional space $\Delta(\Omega)$ of posterior beliefs about $\omega$.
\begin{cor}[Differentially Private Data Publication as Information Design: Binary Data]\label{C:DesignProblem}
If respondents are anonymous and data is binary,  the oblivious mechanism $(S,\sigma)$ solves the designer's problem (\ref{E:OriginalProblem}) if and only if it induces a distribution of posteriors about the population statistic
that solves %
\begin{align}
\max_{\exper\in\Delta(\constraint(\epsilon,\prior))}\{E_\exper \vf(\belief)\text{ s.t. }E_\exper\belief=\prior\}.\label{E:DesignProblem}
\end{align}
\end{cor}

The role played by the respondents' anonymity in Theorem \ref{T:OblivEquiv} is subtle, as is the role of the kind of data that the designer collects about them.
Because 
the data is binary, databases with the same population statistic must be permutations of each other.
Hence, each differs by a single entry from the same number of databases whose population statistic is lower by 1.
Moreover, since these databases also differ from one another only by permutation, a symmetric prior places equal probability on each of them.
Consequently, the bounds that differential privacy places on the posterior likelihood ratios of databases that differ \textit{in one entry} sum to the bounds it places on the posterior likelihood ratios of states that differ \textit{by one} induced by an oblivious mechanism. That is, the constraints that characterize $\constraint(\epsilon,\prior)$ are the projections of the constraints that characterize $\tconstraint(\epsilon,\tprior)$, and so $\constraint(\epsilon,\prior)=\proj\tconstraint(\epsilon,\tprior)$.
This symmetry is crucial for Theorem \ref{T:OblivEquiv}: In Section S.1 of the Online Appendix, we give an example where
respondents are not anonymous, $N>2$, and a non-oblivious mechanism outperforms every oblivious $\epsilon$-DP mechanism.

We emphasize that anonymity leads to Theorem \ref{T:OblivEquiv} \textit{indirectly}: the prior's symmetry is needed not because it eliminates differences between the \textit{probabilities} of permutations of $\theta$, but because it renders the differences between the sets of databases \textit{adjacent} to those permutations irrelevant for differential privacy.\footnote{Specifically, anonymity implies that for each $\omega>0$, there is a joint distribution $\eta$ on $\{\theta: \omega_\theta=\omega\}\times \{\theta: \omega_\theta=\omega-1\}$ with marginals equal to $\tprior(\cdot|\omega)$ and $\tprior(\cdot|\omega-1)$ such that $\supp\eta$ only contains pairs of databases that are adjacent.}
Consequently, it is unnecessary when there are only two respondents: With $N=2$, 
permuting $\theta$ does not change the set of databases that differ from $\theta$ in one entry.

\begin{prop}[Oblivious Mechanisms with Two Respondents]
With binary data and two respondents, $\constraint(\epsilon,\prior)=\proj\tconstraint(\epsilon,\tprior)$: For each $\epsilon$-DP data publication mechanism $(S,\mech)$, there exists an $\epsilon$-DP oblivious data publication mechanism $(S,\sigma)$ that induces the same distribution of posterior
beliefs about the population statistic.\label{P:OblivEquiv2}
\end{prop}

\subsection{Permutation Invariance and Optimal Oblivious Mechanisms}

In general, the key reason for the stark contrast between the conclusions of Theorems \ref{T:OblivMag} and \ref{T:OblivEquiv} is that with binary data, databases with the same population statistic are identical up to permutation, whereas with magnitude data, they are not. 
This suggests that in settings with magnitude data, even though we cannot restrict attention to mechanisms that depend only on the population statistic, we may be able to restrict attention to mechanisms that are invariant under permutation.
Formally, we say that a mechanism is \textit{permutation-invariant} if $\mech(s|\theta)=\mech(s|\theta')$ whenever $\theta$ is a permutation of $\theta'$. 
In the online appendix, we prove the following result.

\begin{prop}[Permutation-Invariant Mechanisms]
If respondents are anonymous, then for each $\epsilon$-DP data publication mechanism, there exists a permutation-invariant data publication mechanism that is also $\epsilon$-DP and induces the same distribution of posterior
beliefs about the population statistic.\label{P:OblivEquivPerm}
\end{prop}

Finally, we note that with binary data and anonymous respondents,
 Theorem \ref{T:OblivEquiv} allows us to reduce Corollary \ref{C:SolutionT} to a set of statements about objects in the lower-dimensional space $\Delta(\Omega)$ of posteriors about the population statistic, and sharpen its characterization of the designer's problem.
 For any $K\subseteq\Delta(\Omega)$, denote the restriction of $\vf$ to $K$ as  $\vf_K:K\to\mathbb{R}\cup\{-\infty\}$, and the concavification of $\vf_K$ as $\ccvf_K:K\to\mathbb{R}$.
Then we have the following corollary to Theorem \ref{T:OblivEquiv}.

\begin{cor}[Characterization of Optimal Oblivious Mechanisms]\label{C:Solution}
Suppose that respondents are anonymous and data is binary.
\begin{enumerate}[i.]
\item The maximized value of the designer's problem (\ref{E:OriginalProblem}) is $\ccvf_{\constraint(\epsilon,\prior)}(\prior)$.
\label{I:Solution_Max}
\item 
There is an oblivious mechanism which solves the designer's problem (\ref{E:OriginalProblem}) and induces a distribution of posteriors about the population statistic $\exper^*\in\Delta(\Delta(\Omega))$ such that
\label{I:Solution_ArgMax}
\begin{enumerate}
\item Each $\belief\in\supp\exper^*$ is an extreme point of $\constraint(\epsilon,\prior)$;
\item 
Each $\belief\in\supp\exper^*$ achieves the privacy bound \eqref{E:DP_posterior} at each value of the population statistic:

 $\left|\log\left(\frac{\belief(\omega)}{\belief(\omega-1)}\right)-\log\left(\frac{\prior(\omega)}{\prior(\omega-1)}\right)\right|=\epsilon$ 
for all $\belief\in\supp\exper^*$ and $\omega\in\Omega\setminus\{0\}$;\label{I:Solution_BoundAttained}
\item 
The support of $\exper^*$ is a linearly independent set of vectors in $\mathbb{R}^{N+1}$; and thus\footnote{
For concreteness, let each belief $\belief\in\Delta(\Omega)$ be represented by the (zero-indexed, i.e., with entries numbered from 0 rather than 1) vector in $\mathbb{R}^{NT+1}$ whose $n$th entry corresponds to the probability $\belief(n)$ it places on $\omega=n$.}
\label{I:Solution_Independence}
\item $\exper^*$ is the unique Bayes-plausible distribution of posteriors with support $\supp\exper^*$.\label{I:Solution_UniqueSupport}
\end{enumerate}
\end{enumerate}
\end{cor}

\section{Monotone Problems and the Geometric Mechanism}
\label{S:SPM}

In a large class of applications where data is binary and respondents are anonymous,
we can move beyond Corollary \ref{C:Solution}'s characterization.
In particular, we focus on settings where the decision maker's set of actions is totally ordered --- and hence can be represented as a subset $A\subseteq\mathbb{R}$ of the real numbers  --- and he views higher actions and higher population statistics as complementary, in that his payoffs form an \textit{interval dominance order (IDO)} family in the sense of \cite{quah2009comparative}:
for all $\omega'<\omega''$ and $a'<a''$,
\[u(a'',\omega')\geq u(a,\omega')\text{ for every }a\in A\cap[a',a'']\Rightarrow u(a'',\omega'')\geq u(a',\omega''),
\]
and this inequality is strict whenever $u(a'',\omega')>u(a',\omega')$.
This condition is more general than supermodularity or the single crossing property. When it is satisfied, we say that the decision maker's problem is \textit{monotone}.\footnote{
Recall that when $A\subseteq\mathbb{R}$, $u:A\times\Omega\to\mathbb{R}$ is supermodular if for each $a\leq a'$ and each $\omega\leq\omega'$, $u(a',\omega')-u(a,\omega')\geq u(a',\omega)-u(a,\omega)$. More generally, $u$ satisfies the single crossing property if for each $a\leq a'$ and each $\omega\leq\omega'$, $u(a',\omega)\geq u(a,\omega)$ implies $ u(a',\omega')\geq u(a,\omega')$, and strictly if the inequality is strict.
}

In this section, we show that these features --- binary data, anonymous respondents, ordered actions, and a DM with a monotone problem --- ensure that the well-known \textit{$\epsilon$-geometric mechanism}   \citep{GhoshRoughgardenSunararajan:Universally:SIAM:2012} always solves the designer's problem (Theorem \ref{T:Geometric}). 
 As described in the introduction, this mechanism adds two-sided geometrically distributed noise to the true population statistic, and so can be formally represented as $(\mathbb{Z},\gmech)$, where\footnote{While the $\epsilon$-geometric mechanism produces more outputs than there are states --- which, by Corollary \ref{C:Solution} (\ref{I:Solution_Independence}), is unnecessary --- it induces the same distribution of posterior beliefs about the population statistic as the \textit{truncated} $\epsilon$-geometric mechanism  $(\Omega,\hat{\sigma}^g_\epsilon)$ \citep{GhoshRoughgardenSunararajan:Universally:SIAM:2012} with signal distribution
\[\hat{\sigma}^g_\epsilon(s|\omega)=\left\{\begin{array}{rl}\left(\frac{1-e^{-\epsilon}}{1+e^{-\epsilon}}\right)e^{-\epsilon|s-\omega|}, & 0<s<N,\\
\left(\frac{1}{1+e^{-\epsilon}}\right)e^{-\epsilon|s-\omega|}, & s\in\{0,N\}.
\end{array}\right.\]
Hence, if a data provider wants to use as few outputs as possible, or just avoid publishing negative outputs, it can truncate the geometrically-distributed noise it adds without changing the mechanism's value to decision makers.
}
\[\gmech(s|\omega)=\left(\frac{1-e^{-\epsilon}}{1+e^{-\epsilon}}\right)e^{-\epsilon|s-\omega|}.\]
Because existing dominance results do not apply to the information structures induced by $\epsilon$-DP mechanisms, we instead introduce a new order on information structures --- \textit{Uniform-Peaked Relative Risk (UPRR) dominance} --- and show that the $\epsilon$-geometric mechanism is UPRR-dominant within the class of $\epsilon$-DP data publication mechanisms. We then complete the argument by showing that a UPRR-dominant information structure is superior in any monotone decision problem with actions on the real line (Theorem \ref{T:ProbDominance}).

\subsection{Optimality of the Geometric Mechanism}

To explain why the geometric mechanism is optimal in these settings, we start by returning to the binary data setting of Figures \ref{F:DPpostT}-\ref{F:DPpost}, with $N=2$, symmetric prior 
$\tprior=
(
\frac{1}{3},\frac{1}{6},
\frac{1}{6},\frac{1}{3})$, and hence $\prior=\proj\tprior=
(\frac{1}{3},\frac{1}{3},\frac{1}{3})
$.
There, our results from Section \ref{S:Oblivious} tell us that when designing a publication mechanism,
\begin{itemize}
\item we need only consider oblivious mechanisms (Theorem \ref{T:OblivEquiv});
    \item we can focus on the distribution of posteriors \textit{about the population statistic} that the mechanism induces (Corollary \ref{C:DesignProblem}); and
    \item it is without loss to consider such distributions supported by a linearly independent set of extreme points of $\constraint(\epsilon,\prior)$, and when we do, the support we choose pins down the distribution (Corollary \ref{C:Solution}).
\end{itemize}
Consequently, the designer 
only ever needs to consider two 
mechanisms in this setting: one oblivious mechanism that induces the set of posteriors in the left panel of Figure \ref{F:twodists}, and another that induces the set of posteriors in the right panel.

\begin{figure}[h!]
    \centering
\begin{tikzpicture}[font=\small]
\def\DPpanelside{4.5}
\pgfmathsetmacro{\DPpanelheight}{sqrt(3)*\DPpanelside/2}
\pgfmathsetmacro{\DPpanelexp}{exp(1)}

\begin{scope}
\coordinate (Lzero) at (0,0);
\coordinate (Lone) at ({\DPpanelside/2},\DPpanelheight);
\coordinate (Ltwo) at (\DPpanelside,0);
\coordinate (Lx) at
    ({\DPpanelside*(1+0.5*\DPpanelexp)/(1+\DPpanelexp+\DPpanelexp*\DPpanelexp)},
     {\DPpanelheight*\DPpanelexp/(1+\DPpanelexp+\DPpanelexp*\DPpanelexp)});
\coordinate (Ly) at
    ({\DPpanelside/2},{\DPpanelheight*\DPpanelexp/(2+\DPpanelexp)});
\coordinate (Lz) at
    ({\DPpanelside*(\DPpanelexp*\DPpanelexp+0.5*\DPpanelexp)/(1+\DPpanelexp+\DPpanelexp*\DPpanelexp)},
     {\DPpanelheight*\DPpanelexp/(1+\DPpanelexp+\DPpanelexp*\DPpanelexp)});
\coordinate (Lw) at
    ({\DPpanelside/2},{\DPpanelheight/(1+2*\DPpanelexp)});
\coordinate (Lpriorstat) at
    ({\DPpanelside/2},{\DPpanelheight/3});
\path[fill=MidnightBlue!16]
    (Lx)--(Ly)--(Lz)--(Lw)--cycle;
\draw[line width=0.65pt] (Lzero)--(Lone)--(Ltwo)--cycle;
\fill[ForestGreen] (Lx) circle (1.35pt);
\fill[DarkRed] (Ly) circle (1.35pt);
\fill[Goldenrod] (Lz) circle (1.35pt);
\fill (Lpriorstat) circle (1.35pt);
\node[below=1pt, text=ForestGreen, font=\large] at (Lx) {$\belief_x$};
\node[above=1pt, text=DarkRed, font=\large] at (Ly) {$\belief_y$};
\node[below=1pt, text=Goldenrod, font=\large] at (Lz) {$\belief_z$};
\node[below=1pt, font=\large] at (Lpriorstat) {$\prior$};
\node[below left=1pt] at (Lzero) {$\omega=0$};
\node[above=1pt] at (Lone) {$\omega=1$};
\node[below right=1pt] at (Ltwo) {$\omega=2$};
\end{scope}

\begin{scope}[shift={(7.6cm,0)}]
\coordinate (Rzero) at (0,0);
\coordinate (Rone) at ({\DPpanelside/2},\DPpanelheight);
\coordinate (Rtwo) at (\DPpanelside,0);
\coordinate (Rx) at
    ({\DPpanelside*(1+0.5*\DPpanelexp)/(1+\DPpanelexp+\DPpanelexp*\DPpanelexp)},
     {\DPpanelheight*\DPpanelexp/(1+\DPpanelexp+\DPpanelexp*\DPpanelexp)});
\coordinate (Ry) at
    ({\DPpanelside/2},{\DPpanelheight*\DPpanelexp/(2+\DPpanelexp)});
\coordinate (Rz) at
    ({\DPpanelside*(\DPpanelexp*\DPpanelexp+0.5*\DPpanelexp)/(1+\DPpanelexp+\DPpanelexp*\DPpanelexp)},
     {\DPpanelheight*\DPpanelexp/(1+\DPpanelexp+\DPpanelexp*\DPpanelexp)});
\coordinate (Rw) at
    ({\DPpanelside/2},{\DPpanelheight/(1+2*\DPpanelexp)});
\coordinate (Rpriorstat) at
    ({\DPpanelside/2},{\DPpanelheight/3});
\path[fill=MidnightBlue!16]
    (Rx)--(Ry)--(Rz)--(Rw)--cycle;
\draw[line width=0.65pt] (Rzero)--(Rone)--(Rtwo)--cycle;
\fill[DarkRed] (Ry) circle (1.35pt);
\fill[Blue] (Rw) circle (1.35pt);
\fill (Rpriorstat) circle (1.35pt);
\node[above=1pt, text=DarkRed, font=\large] at (Ry) {$\belief_y$};
\node[below=1pt, text=Blue, font=\large] at (Rw) {$\belief_w$};
\node[below=1pt, font=\large] at (Rpriorstat) {$\prior$};
\node[below left=1pt] at (Rzero) {$\omega=0$};
\node[above=1pt] at (Rone) {$\omega=1$};
\node[below right=1pt] at (Rtwo) {$\omega=2$};
\end{scope}
\end{tikzpicture}
    \caption{\textbf{Possible solutions to the designer's problem \eqref{E:DesignProblem} in Figures \ref{F:DPpostT}-\ref{F:DPpost}.} In both panels, the shaded region is $\constraint(\epsilon,\prior)$. Only the extreme points in the support under consideration are labeled. Left panel: One possible solution is the unique Bayes-plausible distribution of posteriors about the population statistic with support $\{\belief_x,\belief_y,\belief_z\}$. Right panel: The alternative is the unique Bayes-plausible distribution with support $\{\belief_w,\belief_y\}$.}
    \label{F:twodists}
\end{figure}

These mechanisms provide information about two different aspects of the population statistic $\omega$: inducing the posteriors in the left panel provides information about whether $\omega$ is high or low, while inducing those in the right panel focuses on providing information about whether $\omega$ is central or extreme. 
Which one is better for the decision maker depends on which kind of information is more relevant to his choice of action.

For the class of problems that we focus on in this section --- those where actions are totally ordered, and the DM's problem is monotone --- the answer is straightforward: rightward shifts in the decision maker's belief about the population statistic make higher actions more attractive relative to lower ones. Consequently, it is more valuable for him to know whether the population statistic is high or low than whether it is central or extreme --- and so inducing the posteriors in the left panel is optimal.

This intuition generalizes to binary data settings with more than two respondents.
Observe that 
the beliefs in the left panel of Figure \ref{F:twodists} are precisely those that satisfy
the upper privacy bound $\log\left(\frac{\belief(\omega)}{\belief(\omega-1)}\right)-\log\left(\frac{\prior(\omega)}{\prior(\omega-1)}\right)=\epsilon$ for each value of $\omega$ below some (posterior-specific) cutoff $x$, and the lower privacy bound $\log\left(\frac{\belief(\omega)}{\belief(\omega-1)}\right)-\log\left(\frac{\prior(\omega)}{\prior(\omega-1)}\right)=-\epsilon$ at each value of $\omega$ above the cutoff. 
As we show in Lemma \ref{L:Unimodal}, this property characterizes the posterior beliefs about $\omega$ that are induced by
the $\epsilon$-geometric mechanism.

\begin{lem}[Posteriors Produced by the Geometric Mechanism]\label{L:Unimodal}

With binary data, the $\epsilon$-geometric mechanism induces precisely those oblivious $\epsilon$-DP posteriors 
that
 satisfy the upper privacy bound for each $\omega\leq x$ and the lower privacy bound for each $\omega>x$ for some $x\in\Omega$. That is, letting $\gpost$ denote the distribution of posteriors about the population statistic that is induced by the $\epsilon$-geometric mechanism,
\[\supp\gpost=\left\{\belief\in\ext(\constraint(\epsilon,\prior))\left|\
\frac{\belief(\omega)/\belief(\omega-1)}{\prior(\omega)/\prior(\omega-1)}={\Bigl\lbrace }\begin{array}{rl}
   e^\epsilon,  &\omega\leq x,  \\[-1mm]
    e^{-\epsilon}, &\omega>x, 
\end{array}\right.\text{ for some }x\in\Omega\right\}.
\]
\end{lem}

To understand Lemma \ref{L:Unimodal}, observe that if a decision maker's belief about the population statistic after viewing the output $s$ is $\belief$, then for any $\omega>0$, we have 
\begin{align*}
    \log\left(\frac{\belief(\omega)}{\belief(\omega-1)}\right)-\log\left(\frac{\prior(\omega)}{\prior(\omega-1)}\right)=\log\left(\frac{\sigma_\epsilon^g(s|\omega)}{\sigma_\epsilon^g(s|\omega-1)}\right)=\epsilon(|s-(\omega-1)|-|s-\omega|).
\end{align*}
If $\omega\leq s$, then this simplifies to $\epsilon$, and the upper privacy bound is satisfied; if $\omega>s$, then it simplifies to $-\epsilon$, and the lower privacy bound is satisfied.\footnote{Note that this means the posteriors induced by outputs $s\leq 0$ achieve all $N$ lower privacy bounds, and so (since $\Delta(\Omega)$ has dimension $N$) must be identical; likewise, each output $s\geq N$ induces the same posterior, which achieves all $N$ upper privacy bounds.
}
Hence, the posteriors induced by each output $s$ of the geometric mechanism are collapsed around population statistic $s$ (or the population statistic closest to $s$, if $s\notin[0,N]$) 
as much as the differential privacy constraint will allow. 

This is intuitively optimal when the decision maker wants to minimize a nondecreasing function of the distance between his action and the true population statistic, as in the ``postprocessing'' exercise of \cite{GhoshRoughgardenSunararajan:Universally:SIAM:2012}, since it allows him to be more certain that his chosen action is near the true value of $\omega$. 
But equally intuitively, it is suboptimal when the decision maker is more interested in whether the population statistic is central or extreme --- in which case a mechanism that induces the distribution in the right panel of Figure \ref{F:twodists} (or an $N>2$ analogue) is better.

Theorem \ref{T:Geometric} shows that the crucial feature that leads to the geometric mechanism's optimality is not that the decision maker wants to \textit{match} the population statistic (as in \cite{GhoshRoughgardenSunararajan:Universally:SIAM:2012}), but rather that he wants to take \textit{higher actions when his belief is more skewed toward higher population statistics}. In particular, whenever the decision maker's actions can be ordered from lowest to highest (e.g., in the introduction, from fewer to more buses), and his problem is monotone,
the $\epsilon$-geometric mechanism outperforms any other $\epsilon$-DP oblivious mechanism.

\begin{thm}[Optimality of the Geometric Mechanism for Monotone Decision Problems]\label{T:Geometric}
Suppose that respondents are anonymous and the data is binary.
If the decision maker's actions $A$ are a compact set of real numbers, and he faces a monotone decision problem, 
then the $\epsilon$-geometric mechanism solves the designer's problem (\ref{E:OriginalProblem}).
\end{thm}

\subsection{UPRR-Dominant Information Structures}

Theorem \ref{T:Geometric} does not follow from a direct application of existing results on comparing information structures. 
Oblivious mechanisms which induce 
extreme points of $\constraint(\epsilon,\prior)$ 
are not comparable in the \cite{blackwell1953equivalent} order: By definition, any extreme point of $\constraint(\epsilon,\prior)$ cannot be expressed as a convex combination of the others, and so no distribution supported by those points can be a garbling of another. Moreover, while the posterior beliefs induced by the $\epsilon$-geometric mechanism are ordered by the monotone likelihood ratio property, the posteriors induced by many other oblivious mechanisms 
--- such as  one that induces the posteriors in the right panel of Figure \ref{F:twodists} --- 
need not be ranked in \textit{any} meaningful sense (e.g., first-order stochastic dominance).%
\footnote{Recall that 
a family of distributions $\{\belief_s\}_{s\in S}$ on a finite set $\Omega\subset\mathbb{R}$ is 
MLRP-ordered if
$\belief_{s'}(\omega')/\belief_{s'}(\omega)\geq \belief_s(\omega')/\belief_s(\omega)$ for each $\omega'>\omega$ and $s'>s$. If $\{\belief_s\}_{s\in S}$ are the posterior beliefs induced by a signal $(S,\sigma:\Omega\to\Delta(S))$, they are MLRP-ordered whenever $\sigma$ has the MLRP, i.e., whenever $\sigma(s'|\omega')/\sigma(s'|\omega)\geq \sigma(s|\omega')/\sigma(s|\omega)$ for each $\omega'>\omega$ and $s'>s$ \citep{quah2009comparative, milgrom1981good}.
}
This means we cannot directly appeal to results like those of \cite{lehmann1988comparing}, \cite{quah2009comparative}, \cite{athey2018value}, or \cite{kim2023comparing}, since they do not rank the information structures induced by the publication mechanisms we need to compare. 

Instead, we introduce a new 
notion of dominance among
 information structures in which the $\epsilon$-geometric data publication mechanism dominates all other $\epsilon$-DP oblivious data publication mechanisms. As we explain after Theorem \ref{T:ProbDominance}, UPRR dominance allows us to construct a garbling of the dominant information structure that is more accurate, in \citeauthor{lehmann1988comparing}'s (\citeyear{lehmann1988comparing}) sense, than a signal equivalent to the dominated one. To emphasize the generality of this
 dominance criterion
 (and our results about it), we will often refer to $\omega$ as the \textit{state} in this section.

\begin{defn}[Uniform-Peaked Relative Risk Dominance] For $\exper,\exper'\in\Delta(\Delta(\Omega))$ such that every posterior in $\supp\exper\cup\supp\exper'$ has full support,
 we say that $\exper$ dominates $\exper'$ in the sense of \textit{uniform-peaked relative risk (UPRR)} and write $\exper\succeq_{UPRR}\exper'$ if 
each $\belief\in\supp \exper$ has a \textit{peak} $\omega^*(\belief)\in\Omega$ such that for each $\belief'\in\supp\exper'$, $\belief(\omega)/\belief'(\omega)$ is nondecreasing on $(-\infty,\omega^*(\belief)]$ and nonincreasing on $[\omega^*(\belief),\infty)$.
\end{defn}
In words, one distribution of posteriors on $\Omega\subset\mathbb{N}$ UPRR-dominates another if each posterior in the dominant distribution's support concentrates relatively more probability around one state --- its peak --- than any posterior in the dominated distribution's support.
Equivalently, the \textit{relative risk} $\belief(\omega)/\belief'(\omega)$ of a state $\omega$ under a posterior $\belief$ in the dominant distribution's support against some posterior $\belief'$ in the latter distribution is always increasing in the state below the first posterior's peak $\omega^*(\belief)$ and decreasing above it.%
\footnote{Note that for Bayes-plausible distributions whose posteriors have full support, the UPRR order can be equivalently defined in terms of the signals that generate those distributions: If $\exper$ and $\exper'$ are induced by the data publication mechanisms $(S,\sigma)$ and $(S',\sigma')$, then, after deleting outputs with zero ex ante probability, $\exper\succeq_{UPRR}\exper'$ (and so we can write $\sigma\succeq_{UPRR}\sigma'$) if and only if each $s\in S$ has a peak $\hat{\omega}^*(s)$ such that for each $s'\in S'$, $\sigma(s|\omega)/\sigma'(s'|\omega)$ is nondecreasing on $(-\infty,\hat{\omega}^*(s)]$ and nonincreasing on $[\hat{\omega}^*(s),\infty)$.}
This dominance relation $\succeq_{UPRR}$ is not a partial order or a preorder, but 
as Lemma \ref{L:Transitivity} shows, it is transitive.\footnote{The UPRR relation fails to be either reflexive (consider a distribution that induces one posterior that is single-peaked and another that is uniform) or irreflexive, as is necessary for a strict preorder (as Lemma \ref{L:Geometric_Dominates} shows, the distribution of posteriors $\gpost$ UPRR-dominates itself).}

\begin{lem}[Transitivity of UPRR Dominance]\label{L:Transitivity}
Suppose that $\exper,\exper',\exper''\in\exspace$ are each Bayes-plausible with respect to the same prior and every posterior in their supports has full support.
   If $\exper\succeq_{UPRR}\exper'$ and $\exper'\succeq_{UPRR}\exper''$, then $\exper\succeq_{UPRR}\exper''$.
\end{lem}

\begin{example}\label{Ex:UPRR}

UPRR dominance
allows us to compare the 
two distributions of posteriors described in Figure \ref{F:twodists}.
The distribution $\bpost$ from the right panel of Figure \ref{F:twodists} is supported by $\belief_w$ and $\belief_y$, which are proportional to $(e,1,e)$ and $(1,e,1)$, respectively.
The distribution  $\gpost$ from the left panel of Figure \ref{F:twodists}, on the other hand, is supported by $\belief_x$, $\belief_y$, and $\belief_z$, which are proportional to $(e^2,e,1)$, $(1,e,1)$, and $(1,e,e^2)$, respectively. 
For each posterior $\belief\in\supp\gpost$, the relative risk $\belief(\omega)/\belief'(\omega)$ against either posterior $\belief'\in\supp\bpost$ thus has a single peak as a function of $\omega$: at $0$, if $\belief=\belief_x$; 1, if  $\belief=\belief_y$, or 2, if  $\belief=\belief_z$. 
For instance, if $\belief=\belief_z$, then $(\belief(0)/\belief'(0),\belief(1)/\belief'(1),\belief(2)/\belief'(2))$ is proportional to $(1,1,e^2)$ when $\belief'=\belief_y$ and $(e^{-1},e,e)$ when $\belief'=\belief_w$.
Since these peaks do not depend on $\belief'$, we conclude that $\gpost\succeq_{UPRR}\bpost$.
\end{example}

Lemma \ref{L:Geometric_Dominates} shows that the UPRR dominance of the $\epsilon$-geometric mechanism demonstrated in Example \ref{Ex:UPRR} extends to \textit{all} $\epsilon$-DP data publication mechanisms with any number of respondents.
\begin{lem}[Geometric Mechanisms are UPRR-Dominant]\label{L:Geometric_Dominates}
If $\exper\in\exspace$ is induced by an $\epsilon$-DP oblivious mechanism, then $\gpost\succeq_{UPRR}\exper$,
 where for each $\belief\in\supp\gpost$, $\omega^*(\belief)$ is given by the largest $\omega\in\Omega\setminus\{0\}$ for which $\frac{\belief(\omega)/\belief(\omega-1)}{\prior(\omega)/\prior(\omega-1)}=e^\epsilon$, or 0 if no such $\omega$ exists.
\end{lem}

For intuition, recall that each posterior $\belief\in\Delta(\Omega)$  induced by the geometric mechanism satisfies the upper privacy bound at every value of $\omega$ at or below some $x$ and the lower privacy bound at each higher value of $\omega$. Hence, these posteriors each concentrate probability around their peak $\omega^*(\belief)=x$ as much as the differential privacy constraint \eqref{E:DP_posterior} will allow. This is precisely what is necessary for the distribution they support to UPRR-dominate the distributions of posteriors about the population statistic induced by every other differentially private oblivious mechanism.

We 
conclude with
our general comparative statics result on UPRR dominance. Theorem \ref{T:ProbDominance} shows that among information structures
which induce finitely many posteriors, those
 that are UPRR-dominant are more useful in monotone decision problems with actions on the real line.
 Since the distribution of posterior beliefs about the state induced by the geometric mechanism UPRR-dominates every other distribution of oblivious $\epsilon$-DP posteriors (Lemma \ref{L:Geometric_Dominates}), and the designer's 
 problem can be solved by inducing a distribution of posteriors with finite support
 (Corollary \ref{C:Solution} (\ref{I:Solution_Independence})), Theorem \ref{T:Geometric} follows.

\begin{thm}[UPRR Dominance Implies Dominance in Monotone Decision Problems]\label{T:ProbDominance}
Suppose that $\exper,\exper'\in\exspace$ are Bayes-plausible and have finite support,\footnote{The finite-support assumption is not necessary for the result, but allows us to avoid measurability issues in the proof. It is without loss for our application, since Corollary \ref{C:Solution} (\ref{I:Solution_Independence}) provides an optimal information structure with finite support.} every posterior in $\supp\exper\cup\supp\exper'$ has full support, and $\exper\succeq_{UPRR}\exper'$. If $A\subseteq\mathbb{R}$ is compact, $u:A\times\Omega\to\mathbb{R}$ is continuous in its first argument, and $\{u(\cdot,\omega)\}_{\omega\in\Omega}$ is an IDO family, then
$E_\exper[\max_{a\in A}E_\belief[u(a,\omega)]]\geq E_{\exper'}[\max_{a\in A}E_\belief[u(a,\omega)]]$.
\end{thm}

Observe that when a decision maker's problem is monotone, the actions that are optimal at higher states are higher, so a belief serves him well when it correctly locates the state in its order. Intuitively, then, if two beliefs $\belief$ and $\belief'$ induce the same action, but
$\belief$ places relatively more mass than $\belief'$ on states closer to the state $\omega^*(\belief)$ where that action is optimal,
the states where that action performs poorly --- those whose optimal actions lie far above or far below it --- should arise less often under $\belief$ than under $\belief'$. If this were true for every pair of beliefs $\belief$ and $\belief'$ induced, respectively, by the Bayes-plausible distributions $\exper$ and $\exper'$ --- as is the case when $\exper$ UPRR-dominates $\exper'$ --- it would suggest that any decision maker with a monotone problem would be better off under the information structure $\exper$ than under $\exper'$.

Theorem \ref{T:ProbDominance} confirms this conjecture. Its proof works by connecting UPRR dominance to \citeauthor{lehmann1988comparing}'s (\citeyear{lehmann1988comparing}) accuracy order. From the dominant distribution $\exper$, we construct a garbling: an experiment that reports only the \textit{rank} of the induced posterior's peak $\omega^*(\belief)$. From the dominated distribution $\exper'$, we construct a signal that is equivalent to $\exper'$ for the decision maker: an experiment that reports the rank of the \textit{optimal action} under the induced posterior.
 The key step of the proof shows that UPRR dominance makes the peak-rank experiment more accurate than the action-rank experiment in the sense of \cite{lehmann1988comparing}: as the state rises, high peak ranks become likely more quickly than high action ranks do. Because the decision maker's payoffs form an IDO family, Proposition 9 of \cite{quah2009comparative} then provides a decision rule based on the peak rank that performs at least as well, state by state, as taking the optimal action at each action rank. Since the decision maker can construct the peak rank from his posterior under $\exper$ and follow that policy, $\exper$ is at least as valuable to him as $\exper'$. Unlike existing results, this argument does not require either information structure to induce MLRP-ordered posteriors, and the posteriors induced by $\exper'$ need not be ordered in any sense.

\section{Conclusion}

Our analysis introduces the tools of information design to the problem of differentially private data publication.
First, Corollaries \ref{C:SolutionT} and  \ref{C:Solution} show how to select an optimal publication mechanism using concavification results from the information design literature. Even if a data provider does not wish to derive an optimal mechanism explicitly, these results establish criteria for checking whether any given mechanism is potentially optimal by counting the privacy bounds that it attains. 

Second, most differentially private data publication mechanisms used in practice are \textit{oblivious}: they add conditionally independent noise to the true population statistic, rather than publishing output in a way that depends more generally on the database.
As we show, using such mechanisms is without loss when 
the data is binary and
the designer and data users view respondents as interchangeable (e.g., if they are i.i.d.), but can be suboptimal otherwise.  With binary data, if the identity of a respondent carries information about his or her type,
the designer may be able to exploit that information to design a mechanism that produces more useful output while maintaining differential privacy.
And with magnitude data, the designer can exploit the fact that (unlike with binary data) differential privacy 
restricts beliefs differently at different databases with the same population statistic.

Finally, we show that among oblivious mechanisms, the commonly used \textit{geometric mechanism} is optimal whenever data users view their actions and the population statistic of interest as complementary, in the sense that their decision problems are monotone. We do so by using a novel result that compares the value of information structures in monotone decision problems.

Outside of the binary, monotone case considered by Theorem \ref{T:Geometric}, the mechanisms that solve the designer's problem are not necessarily easy for humans to interpret. For instance, the optimal mechanism in the magnitude data case from  Example \ref{Ex:School} is a vector with over one hundred thousand parameters. 
We thus have very little to say qualitatively about what they might look like.
But the practical message of our paper for settings with magnitude data, or where data users face non-monotone problems, is not a set of intuitions about when one mechanism is better than another.
Rather, it is that data users can benefit significantly when data publishers replace off-the-shelf publication mechanisms with those constructed by computing the solution to an optimization problem --- one that, like other information design problems, can be formulated as a linear program.
With binary data and anonymous respondents, we can greatly reduce the dimensionality of that problem by restricting to oblivious mechanisms (Theorem \ref{T:OblivEquiv}), allowing its solution to be computed more quickly; with magnitude data, on the other hand, we cannot (Theorem \ref{T:OblivMag}; Example \ref{Ex:School}).

While we focus on the standard differential privacy criterion, other variants exist. Two popular ones are \textit{approximate} or $(\epsilon,\delta)$\textit{-differential privacy}, which adds a tolerance $\delta$ to the bound on the probability of any set of outputs conditional on databases that differ in one entry;\footnote{Formally, $(S,\mech)$ is $(\epsilon,\delta)$-differentially private if $\sum_{s\in E}\mech(s|\theta)\leq e^\epsilon\sum_{s\in E}\mech(s|\theta')+\delta$ for every $E\subseteq S$ and any databases $\theta,\theta'\in\Theta$ with $\theta_{-i}=\theta'_{-i}$ for some $i$.} and \textit{zero-concentrated differential privacy} (zCDP), which instead bounds the R\'enyi divergences between the distributions of a mechanism's output at such databases, and which the Census Bureau adopted for the 2020 Decennial Census \citep{bun2016concentrated}. Our conclusions about oblivious mechanisms hold for both: in subsequent work, \cite{pan2025differentially} show that
Theorem \ref{T:OblivEquiv} extends to $(\epsilon,\delta)$-differential privacy, and in Section S.5 of the Online Appendix, we show that Theorems \ref{T:OblivMag} and \ref{T:OblivEquiv} both extend to zCDP. Neither criterion, however, amounts to a constraint on the set of posterior beliefs that the designer can induce: they also show this for approximate differential privacy, while zCDP instead constrains the entire distribution of posteriors that the mechanism induces. In each case, the designer thus faces a constrained information design problem more complex than \eqref{E:DesignProblemT}.\footnote{\cite{pan2025differentially} also generalize our results by considering misalignment between the preferences of the data users and those of the designer.}

Several applications suggest generalizations of our model.
First, in many settings, data users interact with one another, instead of making independent decisions. For instance, when considering whether to attend a gathering during a pandemic, individuals may not only consider the information published about disease prevalence in their community, but also the attendance decisions of others.
It may thus be worthwhile to identify the circumstances in which the geometric mechanism remains optimal when data users have monotone decision problems that are also affected by the actions taken by others. %

 Second, many data providers wish to publish information about 
multiple characteristics, each of which might affect the data user's decision problem through the number of respondents that have it. Such \textit{multidimensional} data publication problems have been considered by \cite{BrennerNissim:Impossibility:SIAM:2014}, whose counterexample shows that in contrast to the one-dimensional setting of \cite{GhoshRoughgardenSunararajan:Universally:SIAM:2012}, the geometric mechanism is not universally optimal for all data users who minimize symmetric loss functions. To our knowledge, no characterization of optimal publication mechanisms in multidimensional settings has appeared in the privacy literature.
While it is straightforward to show that a characterization result analogous to Corollary \ref{C:SolutionT} can be given in such settings, the circumstances under which oblivious mechanisms are without loss of generality are less clear. We leave 
this question
to future work.

\nocite{Mathematica}
\setlength{\bibsep}{0pt plus 0.3ex}
\bibliographystyle{ecta}
\def\bibfont{\small}
\bibliography{dp_as_bp}
\appendix
\section*{Appendix}
\linespread{1.15}
\small

\small
\section{Proofs}\label{S:Proofs}

\noindent\textit{Notation.} For a set $S$, denote its relative interior by $\ri(S)$ and its relative boundary by $\rbd(S)$.
For vectors in $\mathbb{R}^{NT+1}$, we adopt the convention of \textit{zero-indexing}: we number the entries of vectors from 0 to $NT$ rather than from 1 to $NT+1$.
We write $\textbf{1}$ to denote a (conformable) vector of ones.
For $S\subset\mathbb{Z}_+$, we write $\textbf{1}_S$ to denote the indicator vector for $S$, i.e., the vector (of conformable dimension) with 1 in its $i$th entry if $i\in S$ and 0 in its $i$th entry otherwise. Likewise, we write $\textbf{1}_n=\textbf{1}_{\{n\} }$.
For $\theta\in\tspace$, we write $\textbf{1}_{\theta}=\textbf{1}_{m}\in\mathbb{R}^{(T+1)^N}$, where $m$ is the representation of $\theta$ in base $(T+1)$.

\subsubsection*{Characterization of the Designer's Problem and (Obliviously) Differentially Private Posteriors}
 
 \begin{lem}\label{L:DP}
 The distribution $\texper\in\texspace$ of posterior beliefs can be induced with an $\epsilon$-DP mechanism $(S,\mech)$ if and only if $\texper$ is Bayes-plausible, has countable support, and $\supp\texper\subset\tconstraint(\epsilon,\tprior)$.
 \end{lem}
 
 \begin{proof}
 By \cite{KamenicaBayesianPersuasion2011},
 $\texper$ can be induced with some mechanism $(S,\mech)$ if and only if it is Bayes-plausible; this is only possible with countable $S$ (as we require) if $\supp\texper$ is countable. From Bayes' rule \eqref{E:Bayes} and the definition of differential privacy \eqref{E:DP}, $(S,\mech)$ is $\epsilon$-DP if and only if $\supp\texper\subset\tconstraint(\epsilon,\tprior)$ as well.
 \end{proof}

\begin{lem}[Characterization of Differentially Private Posteriors]\label{L:tconstraint}
For each $\epsilon>0$,
\begin{enumerate}[i.]
\item $\tconstraint(\epsilon,\tprior)$ lies in the relative interior of $\Delta(\tspace)$. %
\label{I:tconstraintintsimplex}
\item $\tprior$ lies in the relative interior of $\tconstraint(\epsilon,\tprior)$.\label{I:tconstraintintprior}
\item $\tconstraint(\epsilon,\tprior)$ is a closed convex polyhedron in $\mathbb{R}^{(T+1)^N}$.\label{I:tconstraintconvex}
\item $\tconstraint(\epsilon,\tprior)$ has finitely many extreme points and $\tconstraint(\epsilon,\tprior)=\conv(\ext(\tconstraint(\epsilon,\tprior)))$.\label{I:tconstraintextreme}
\end{enumerate}
\end{lem}

\begin{proof}

(\ref{I:tconstraintintsimplex}): If $\tbelief\in\rbd(\Delta(\tspace))$, then $\tbelief(\theta)=0$ for some $\theta\in\tspace$. 
We cannot have $\tbelief(\theta)=0$ for all $\theta\in\tspace$, so it follows that there must exist $\theta',\theta''\in\tspace$, with  $\theta''_{-i}=\theta'_{-i}$ for some $i$, such that $\tbelief(\theta')=0$ and $\tbelief(\theta'')>0$. Then $\left|\log(\tbelief(\theta'')/\tbelief(\theta'))\right|=\infty$, so $\tbelief\notin\tconstraint(\epsilon,\tprior)$.

(\ref{I:tconstraintintprior}): 
Since $\log(x)$ is a continuous function on $(0,\infty)$, for each $\theta,\theta'\in\tspace$, $\log(\tbelief(\theta)/\tbelief(\theta'))$ is continuous in 
$\tbelief$ on $\ri(\Delta(\tspace))$. Then for each $\theta,\theta'\in\tspace$ with  $\theta_{-i}=\theta'_{-i}$ for some $i$, the set \[\left\{\tbelief\in\ri(\Delta(\tspace))\mid\log\left(\frac{\tbelief(\theta)}{\tbelief(\theta')}\right)\in\left(\log\left(\frac{\tprior(\theta)}{\tprior(\theta')}\right)-\epsilon,\log\left(\frac{\tprior(\theta)}{\tprior(\theta')}\right)+\epsilon\right)\right\}\] is open in $\ri(\Delta(\tspace))$ and hence $\aff(\Delta(\tspace))$. 
Then so is their intersection, \[\hat{K}=\left\{\tbelief\in\ri(\Delta(\tspace))\mid-\epsilon<\log\left(\frac{\tbelief(\theta)}{\tbelief(\theta')}\right)-\log\left(\frac{\tprior(\theta)}{\tprior(\theta')}\right)<\epsilon\ \forall\theta,\theta'\text{ s.t. }\exists i,\theta_{-i}=\theta'_{-i} \right\}.\] 
Then since $\hat{K}$ is open and $\tprior\in\hat{K}\subset \tconstraint(\epsilon,\tprior)$, we have
$\tprior\in\ri(\tconstraint(\epsilon,\tprior))$, as desired.

 (\ref{I:tconstraintconvex}): We have
 \small
\begin{align*}
\tconstraint(\epsilon,\tprior)
&=\left\{\tbelief\in\Delta(\tspace)\left| e^{-\epsilon}\leq\frac{\tbelief(\theta)/\tprior(\theta)}{\tbelief(\theta')/\tprior(\theta')}\leq e^\epsilon\ \forall\theta,\theta'\text{ s.t. }\exists i,\theta_{-i}=\theta'_{-i}\right.\right\}\\
&=\left\{\tbelief\in\Delta(\tspace)\left| e^{-\epsilon}\tbelief(\theta')\frac{\tprior(\theta)}{\tprior(\theta')}-\tbelief(\theta)\leq 0\leq e^\epsilon\tbelief(\theta')\frac{\tprior(\theta)}{\tprior(\theta')}-\tbelief(\theta)\ \forall\theta,\theta'\text{ s.t. }\exists i,\theta_{-i}=\theta'_{-i}\right.\right\}
\end{align*}
\normalsize
Thus, $\tconstraint(\epsilon,\tprior)$ %
is the intersection of finitely many half-spaces described by the inequalities 
\begin{align}
\tbelief\cdot \left( \textbf{1}_{\theta}-e^{-\epsilon}\begin{array}{@{}c@{}}\frac{\tprior(\theta)}{\tprior(\theta')}\end{array}\textbf{1}_{\theta'}\right)&\geq 0,& 
\tbelief\cdot \left(e^{\epsilon}\begin{array}{@{}c@{}}\frac{\tprior(\theta)}{\tprior(\theta')}\end{array}\textbf{1}_{\theta'}-\textbf{1}_{\theta}\right)&\geq 0,& \forall\theta,\theta'&\text{ s.t. }\exists i,\theta_{-i}=\theta'_{-i};\nonumber
\\
\tbelief\cdot 1&\geq 1; &
\tbelief\cdot 1&\leq 1. & &\label{E:tconstraintsim}
\end{align}
(From (\ref{I:tconstraintintsimplex}), 
$\tbelief(\theta)>> 0$ for each $\tbelief\in\tconstraint(\epsilon,\tprior)$, and so the nonnegativity constraints in the characterization of $\Delta(\tspace)$ can be  omitted from the description of $\tconstraint(\epsilon,\tprior)$.) 
Then by definition, $\tconstraint(\epsilon,\tprior)$ is a  closed convex polyhedron.

(\ref{I:tconstraintextreme}): By Lemma \ref{L:tconstraint} (\ref{I:tconstraintconvex}), $\tconstraint(\epsilon,\tprior)$ is a convex closed polyhedron; since it is a subset of the simplex $\Delta(\tspace)$, it is bounded and thus compact. 
The claim then follows by the Minkowski-Weyl theorem.
\end{proof}

\begin{cor}[Characterization of Oblivious Differentially Private Posteriors]\label{C:constraint}
For each $\epsilon>0$,
\begin{enumerate}[i.]
\item $\constraint(\epsilon,\prior)$ lies in the relative interior of $\Delta(\Omega)$. %
\label{I:constraintintsimplex}
\item $\prior$ lies in the relative interior of $\constraint(\epsilon,\prior)$.\label{I:constraintintprior}
\item $\constraint(\epsilon,\prior)$ is a closed convex polyhedron in $\mathbb{R}^{NT+1}$.\label{I:constraintconvex}
\item $\constraint(\epsilon,\prior)$ has finitely many extreme points and $\constraint(\epsilon,\prior)=\conv(\ext(\constraint(\epsilon,\prior)))$.\label{I:constraintextreme}
\end{enumerate}
\end{cor}
\begin{proof}
Follows identically to Lemma \ref{L:tconstraint}.
\end{proof}

\begin{proofof}{\bf Proposition \ref{P:DPObliv} (Differential Privacy for Oblivious Mechanisms)}
((\ref{I:DPObliv:DP})$\Rightarrow$(\ref{I:DPObliv:DPObliv})) 
Fix $\omega\in\{1,\ldots,N\hitype\}$ and $1\leq t\leq\min\{\hitype,\omega\}$. Choose $i\in\{1,\ldots,N\}$ and $\theta_{-i}=\theta'_{-i}\in\{0,\ldots,\hitype\}^{N-1}$ such that
$
\sum_{j\neq i}\theta_j=\omega-\min\{\hitype,\omega\}.
$
Such a choice is possible because $0\leq\omega-\min\{\hitype,\omega\}\leq(N-1)\hitype$. Set $\theta_i=\min\{\hitype,\omega\}$ and $\theta_i'=\min\{\hitype,\omega\}-t$. Then $\omega_\theta=\omega$, $\omega_{\theta'}=\omega-t$, and $\theta_{-i}=\theta'_{-i}$. Since $(S,\sigma)$ is $\epsilon$-DP, for each $s\in S$,
\[
\left|\log\left(\frac{\sigma(s|\omega)}{\sigma(s|\omega-t)}\right)\right|
=
\left|\log\left(\frac{\sigma(s|\omega_\theta)}{\sigma(s|\omega_{\theta'})}\right)\right|
\leq\epsilon.
\]
Thus, (\ref{I:DPObliv:DPObliv}) follows. 

((\ref{I:DPObliv:DPObliv})$\Rightarrow$(\ref{I:DPObliv:DP})) If $\theta,\theta'\in\{0,\ldots,T\}^{N}$ 
are such that $\theta_{-i}=\theta'_{-i}$ for some $i$,
then either $\theta=\theta'$, in which case (\ref{E:DP}) holds trivially, or 
$|\omega_\theta-\omega_{\theta'}|=t$ for some $1\leq t\leq\min\{\hitype,\max\{\omega_\theta,\omega_{\theta'}\}\}$, in which case (\ref{I:DPObliv:DPObliv}) implies that for each $s\in S$, $\left|\log\left({\sigma(s|\omega_{\theta})}/ {\sigma(s|\omega_{\theta'})}\right)\right|=\left|\log\left({\sigma(s|\omega_{\theta'})}/ {\sigma(s|\omega_{\theta})}\right)\right|\leq\epsilon,$  and hence, since $(S,\sigma)$ is oblivious, (\ref{E:DP}).

((\ref{I:DPObliv:DPObliv})$\Leftrightarrow$(\ref{I:DPObliv:posterior})) Follows from Bayes' rule, since
\begin{align*}\frac{\belief(\omega')}{\belief(\omega)}=
\left.\frac{\sigma(s|\omega')\prior(\omega')}{\sum_{x\in\Omega}\sigma(s|x)\prior(x)}\right/\frac{\sigma(s|\omega)\prior(\omega)}{\sum_{x\in\Omega}\sigma(s|x)\prior(x)}=
\frac{\sigma(s|\omega')}{\sigma(s|\omega)}\frac{\prior(\omega')}{\prior(\omega)}.
\end{align*}
\end{proofof}

\begin{lem}[Characterization of Extreme Points of $\tconstraint(\epsilon,\tprior)$]\label{L:extremalT} \hspace{0pt}
\begin{enumerate}[i.]
\item If the support of a distribution
$\texper$
of $\epsilon$-DP posteriors 
contains a belief that is not an extreme point of $\tconstraint(\epsilon,\tprior)$, then there is some distribution of posterior beliefs $\hat{\texper}$
with $\supp\hat{\texper}\subseteq\ext(\tconstraint(\epsilon,\tprior))$
that is more informative than $\texper$ (in the Blackwell sense).
\label{I:BoundAttainedT}
\item Each extreme point of $\tconstraint(\epsilon,\tprior)$ attains at least $(T+1)^N-1$ privacy bounds:
For each $\tbelief\in\ext(\tconstraint(\epsilon,\tprior))$,
$\left|\log\left(\frac{\tbelief(\theta)}{\tbelief(\theta')}\right)-\log\left(\frac{\tprior(\theta)}{\tprior(\theta')}\right)\right|=\epsilon$ for at least $(T+1)^N-1$ distinct combinations $\{\theta,\theta'\}$ that have $\theta_{-i}=\theta'_{-i}$ for some $i$.
\label{I:VerticesT}
\end{enumerate}
\end{lem}

\begin{proof}
(\ref{I:BoundAttainedT}): By Lemma \ref{L:tconstraint} (\ref{I:tconstraintextreme}),
$\tconstraint(\epsilon,\tprior)$ has finitely many extreme points and $\tconstraint(\epsilon,\tprior)=\conv(\ext(\tconstraint(\epsilon,\tprior)))$.
Fix a triangulation of $\tconstraint(\epsilon,\tprior)$ using only its extreme points as vertices. For each $\tbelief'\in\ext(\tconstraint(\epsilon,\tprior))$, let $\lambda(\cdot,\tbelief')\in[0,1]$ be its piecewise barycentric coordinate, set equal to zero on simplices that do not contain $\tbelief'$. These coordinates agree on common faces and hence are continuous and Borel measurable on $\tconstraint(\epsilon,\tprior)$. 
Then for each $\tbelief\in\supp\texper$, we can write $\tbelief=\sum_{\tbelief'\in\ext(\tconstraint(\epsilon,\tprior))}\lambda(\tbelief,\tbelief')\tbelief'$ with $\sum_{\tbelief'\in\ext(\tconstraint(\epsilon,\tprior))}\lambda(\tbelief,\tbelief')=1$.
Then $\lambda$ is a mean-preserving stochastic transformation on $\tconstraint(\epsilon,\tprior)$ (in the sense of \cite{blackwell1953equivalent}); consequently, by \cite{blackwell1953equivalent} Theorem 2, 
$\hat{\texper}$ defined by
$\hat{\texper}(\tbelief')=
\int_{\tconstraint(\epsilon,\tprior)}\lambda(\tbelief,\tbelief')d\texper(\tbelief)$
is more informative than $\texper$.

(\ref{I:VerticesT}): 
By Lemma \ref{L:tconstraint} (\ref{I:tconstraintconvex}), $\tconstraint(\epsilon,\tprior)\subset\mathbb{R}^{(T+1)^N}$ is the polyhedron described by 
\eqref{E:tconstraintsim}. 
At any $\tbelief\in\ext(\tconstraint(\epsilon,\tprior))$, the set of normal vectors $Z_\tbelief$ of the inequalities in 
\eqref{E:tconstraintsim} that bind spans $\mathbb{R}^{(T+1)^N}$: Suppose not. Then there exists $0\neq z\in(\spn Z_\tbelief)^\perp$. Since there are finitely many nonbinding inequalities in \eqref{E:tconstraintsim}, there exists $\delta>0$ such that $\{\tbelief+\delta z,\tbelief-\delta z\}\subset\tconstraint(\epsilon,\tprior)$, a contradiction since $\tbelief\in\ext(\tconstraint(\epsilon,\tprior))$. Then $Z_\tbelief\subseteq\mathbb{R}^{(T+1)^N}$ contains $(T+1)^N$ linearly independent vectors.

For any $\theta\in\tspace$ and $\tbelief\in\tconstraint(\epsilon,\tprior)$, $\tbelief(\theta)\neq 0$ by Lemma \ref{L:tconstraint} (\ref{I:tconstraintintsimplex}), and so for any $\theta,\theta'\in\tspace$, we cannot have both %
\begin{align}
    \tbelief(\theta)
=e^{-\epsilon}\frac{\tprior(\theta)}{\tprior(\theta')}\tbelief(\theta')
\quad\text{and}\quad
\tbelief(\theta)
=e^\epsilon\frac{\tprior(\theta)}{\tprior(\theta')}\tbelief(\theta'),\label{E:tconstraintbind}
\end{align}
And since both $\tbelief\cdot \textbf{1}\geq 1$ and $\tbelief\cdot (-\textbf{1})\geq -1$ must bind, and $\textbf{1}$ and $-\textbf{1}$ are linearly dependent, it follows that we must have $\left|\log\left(\frac{\tbelief(\theta)}{\tbelief(\theta')}\right)-\log\left(\frac{\tprior(\theta)}{\tprior(\theta')}\right)\right|=\epsilon$ for at least $(T+1)^N-1$ distinct combinations $\{\theta,\theta'\}$ that have $\theta_{-i}=\theta'_{-i}$ for some $i$.
\end{proof}

\begin{cor}[Characterization of Extreme Points of $\constraint(\epsilon,\prior)$]\label{C:extremal} \hspace{0pt}
Suppose that data is binary.
\begin{enumerate}[i.]
 \item If the support of a distribution
$\exper$
of oblivious $\epsilon$-DP posteriors 
contains a belief that is not an extreme point of $\constraint(\epsilon,\prior)$, then there is some distribution of posterior beliefs about the population statistic $\hat{\exper}\in\exspace$
with $\supp\hat{\exper}\subseteq\ext(\constraint(\epsilon,\prior))$
that is more informative than $\exper$ (in the Blackwell sense).
\label{I:BoundAttained}
\item The extreme points of $\constraint(\epsilon,\prior)$ are precisely those for which the privacy bound is attained at all $\omega$: $\ext(\constraint(\epsilon,\prior))=
\left\{\belief\in\Delta(\Omega) \mid 
\left|\log\left(\frac{\belief(\omega)}{\belief(\omega-1)}\right)-\log\left(\frac{\prior(\omega)}{\prior(\omega-1)}\right)\right|
=\epsilon \ \forall\omega\in\Omega\setminus\{0\}\right\}$.\label{I:Vertices}
\end{enumerate}
\end{cor}

\begin{proof}
(\ref{I:BoundAttained}):
Follows identically to Lemma \ref{L:extremalT} (\ref{I:BoundAttainedT}).

(\ref{I:Vertices}): 
By Corollary \ref{C:constraint} (\ref{I:constraintconvex}), $\constraint(\epsilon,\prior)$ is a closed convex polyhedron described by the inequalities 
\begin{align}
\belief\cdot \left( \textbf{1}_{\omega}-e^{-\epsilon}\begin{array}{@{}c@{}}\frac{\prior(\omega)}{\prior(\omega-1)}\end{array}\textbf{1}_{\omega-1}\right)&\geq 0,& 
\belief\cdot \left(e^{\epsilon}\begin{array}{@{}c@{}}\frac{\prior(\omega)}{\prior(\omega-1)}\end{array}\textbf{1}_{\omega-1}-\textbf{1}_{\omega}\right)&\geq 0,& \omega&\in\Omega\setminus\{0\};\nonumber
\\
\belief\cdot 1&\geq 1; &
\belief\cdot 1&\leq 1. & &\label{E:constraintsim}
\end{align}

At any $\belief\in\ext(\constraint(\epsilon,\prior))$, the set of normal vectors $Z_\belief$ of the inequalities in 
\eqref{E:constraintsim} that bind spans $\mathbb{R}^{N+1}$: Suppose not. Then there exists $0\neq z\in(\spn Z_\belief)^\perp$. Since there are finitely many nonbinding inequalities in \eqref{E:constraintsim}, there exists $\delta>0$ such that $\{\belief+\delta z,\belief-\delta z\}\subset\constraint(\epsilon,\prior)$, a contradiction since $\belief\in\ext(\constraint(\epsilon,\prior))$. Then $Z_\belief\subseteq\mathbb{R}^{N+1}$ contains $N+1$ linearly independent vectors.

Now by definition, $\belief\cdot \textbf{1}= 1$ for each $\belief\in\constraint(\epsilon,\prior)$.
For any $\omega\in\Omega$ and $\belief\in\constraint(\epsilon,\prior)$, $\belief(\omega)\neq 0$ by Corollary \ref{C:constraint} (\ref{I:constraintintsimplex}), and so for any $\omega\in\Omega\setminus\{0\}$, we cannot have both %
\begin{align}
\belief\cdot \left( \textbf{1}_{\omega}-e^{-\epsilon}\frac{\prior(\omega)}{\prior(\omega-1)}\textbf{1}_{\omega-1}\right)&=0 & \text{ and } & &
\belief\cdot \left(e^{\epsilon}\frac{\prior(\omega)}{\prior(\omega-1)}\textbf{1}_{\omega-1}-\textbf{1}_{\omega}\right)&= 0.\label{E:constraintbind}
\end{align}
 Hence, if $\belief\in \ext(\constraint(\epsilon,\prior))$, 
then for each $\omega\in\Omega\setminus\{0\}$, exactly one 
equation in (\ref{E:constraintbind})
must hold to reach a total of $N+1$ linearly independent binding inequalities from %
(\ref{E:constraintsim}); it follows that $|\log\left(\belief(\omega)/\belief(\omega-1)\right)-\log\left(\prior(\omega)/\prior(\omega-1)\right)|=\epsilon$ for each $\omega\in\Omega\setminus\{0\}$.

Conversely, suppose that $\belief\in\Delta(\Omega)$ has $|\log\left(\belief(\omega)/\belief(\omega-1)\right)-\log\left(\prior(\omega)/\prior(\omega-1)\right)|=\epsilon$ for each $\omega\in\Omega\setminus\{0\}$. Then for each $\omega\in\Omega\setminus\{0\}$, 
one of the equations in (\ref{E:constraintbind})
must hold; 
then the set of normal vectors $Z_\belief$ of the inequalities in 
\eqref{E:constraintsim} that bind contains, up to sign, $\{\textbf{1}\}\cup\{e^{\epsilon}\frac{\prior(\omega)}{\prior(\omega-1)}\textbf{1}_{\omega-1}-\textbf{1}_{\omega}|\omega\in S\}\cup\{e^{-\epsilon}\frac{\prior(\omega)}{\prior(\omega-1)}\textbf{1}_{\omega-1}-\textbf{1}_{\omega}|\omega\in(\Omega\setminus\{0\})\setminus S\}$ for some $S\subseteq \Omega\setminus\{0\}$. For any $S\subseteq \Omega\setminus\{0\}$, this set of $N+1$ vectors is linearly independent; it follows that $\spn Z_\belief=\mathbb{R}^{N+1}$. 

Now suppose that $\belief=\alpha x+(1-\alpha)y$ for some $\alpha\in(0,1)$ and $x,y\in\constraint(\epsilon,\prior)$ with $x\neq y$. 
For each $z\in Z_\belief\setminus\{\textbf{1}\}$, we have $z\cdot\belief=0$, $z\cdot y\geq 0$, and $z\cdot x\geq 0$; it follows that $z\cdot x=z\cdot y=0$. And for $z=\textbf{1}$, we have $z\cdot x=z\cdot y=z\cdot\belief=1$ since $x,y,\belief\in\constraint(\epsilon,\prior)$. Since $\spn Z_\belief=\mathbb{R}^{N+1}$, we have $x=y=\belief$, a contradiction. It follows that $\belief\in\ext(\constraint(\epsilon,\prior))$.
\end{proof}

 \begin{proofof}{\bf Proposition \ref{P:DesignProblemT} (Differentially Private Data Publication as Information Design)}
  The ``if''  part follows immediately from Lemma \ref{L:DP}. For the ``only if'' part, observe that since $\tvf$ is convex, it follows from Lemma \ref{L:extremalT}\eqref{I:BoundAttainedT} and Lemma \ref{L:tconstraint}\eqref{I:tconstraintextreme} that it suffices to consider 
  finitely supported $\texper$
  in \eqref{E:DesignProblemT}; the claim then follows from Lemma \ref{L:DP}.\hfill $\square $\smallskip
\end{proofof}

\begin{lem}[Distributions of Posteriors with Affine and Linear Independent Supports]\label{L:independenceT}\hspace{0pt}
\begin{enumerate}[i.]
    \item If $\texper\in\texspace$ is Bayes-plausible and $\supp\texper$ is finite,\footnote{The finiteness assumption is unnecessary here, but is adopted for simplicity.} there is a Bayes-plausible $\texper'\in\texspace$ such that the posteriors in $\supp \texper'$ are affinely independent and $\supp\texper'\subseteq\supp\texper$.\label{I:independenceT:supportindependence}
    \item If $\{\tbelief_j\}_{j=1}^J\subset\Delta(\tspace)$ is affinely independent, it is linearly independent.\label{I:independenceT:AffLinDep}
    \item If $\texper\in\texspace$ is Bayes-plausible and $\supp\texper$ is linearly independent, then $\supp\texper'\neq\supp\texper$ for any Bayes-plausible $\texper'\in\texspace, \texper'\neq\texper$.\label{I:independenceT:uniqueindependent}
\end{enumerate}
\end{lem}

\begin{proof}
\eqref{I:independenceT:supportindependence}: Since $\supp\texper$ is finite and $\texper$ is Bayes-plausible, we have $\sum_{\tbelief\in\supp\texper}\texper(\tbelief)\tbelief=\tprior$, and so $\tprior\in\conv(\supp\texper)$. 
Then by Carath\'{e}odory's theorem (e.g., \cite{rockafellar} Theorem 17.1) there is some affinely independent $\{\tbelief_j\}_{j=1}^J\subseteq\supp\texper$ such that $\tprior\in\conv(\{\tbelief_j\}_{j=1}^J)$. 
Then there exist %
$\{\lambda_j\}_{j=1}^J\subset[0,1]$ such that $\sum_{j=1}^J\lambda_j=1$ and $\sum_{j=1}^J\lambda_j\tbelief_j=\tprior$. The claim follows by letting $\texper'(\tbelief_j)=\lambda_j$ for each $j\in\{1,\ldots,J\}$.

\eqref{I:independenceT:AffLinDep}: We prove the contrapositive.
Suppose that $\{\tbelief_j\}_{j=1}^J$ is linearly dependent. Then there exist $\{\lambda_j\}_{j=1}^J\subset\mathbb{R}$ such that $\sum_{j=1}^J\lambda_j\tbelief_j=0$ and $\lambda_j\neq 0$ for some $j$. 
Then $0=\sum_{\theta\in\tspace}\sum_{j=1}^J\lambda_j\tbelief_j(\theta)=\sum_{j=1}^J\lambda_j\sum_{\theta\in\tspace}\tbelief_j(\theta)$.
Since $\{\tbelief_j\}_{j=1}^J\subset\Delta(\tspace)$, $\sum_{\theta\in\tspace}\tbelief_j(\theta)=1$ for each $j\in\{1,\ldots,J\}$. It follows that $0=\sum_{j=1}^J\lambda_j$ and thus $\{\tbelief_j\}_{j=1}^J$ is affinely dependent.

\eqref{I:independenceT:uniqueindependent}: We proceed by contradiction: Suppose that $\supp\texper'=\supp\texper$ for some Bayes-plausible $\texper'\neq\texper$. Since $\supp\texper$ is linearly independent, it must be finite.
Since $\texper$ and $\texper'$ are both Bayes-plausible, we must have $\sum_{\tbelief\in\supp\texper}\texper(\tbelief)\tbelief=\sum_{\tbelief\in\supp\texper}\texper'(\tbelief)\tbelief=\tprior$. Then we have
$
\sum_{\tbelief\in\supp\texper}(\texper(\tbelief)-\texper'(\tbelief))\tbelief=0, 
$
with $\texper(\tbelief)-\texper'(\tbelief)\neq 0$ for some $\tbelief\in\supp\texper$, since $\texper\neq\texper'$. Then $\supp\texper$ is linearly dependent, a contradiction.
\end{proof}

\begin{cor}[Distributions of Posteriors About $\omega$ with Affine and Linear Independent Supports]\label{C:independence}\hspace{0pt}
\begin{enumerate}[i.]
    \item If $\exper\in\exspace$ is Bayes-plausible and $\supp\exper$ is finite, there is a Bayes-plausible $\exper'\in\exspace$ such that the posteriors in $\supp \exper'$ are affinely independent and $\supp\exper'\subseteq\supp\exper$.\label{I:independence:supportindependence}
    \item If $\{\belief_j\}_{j=1}^J\subset\Delta(\Omega)$ is affinely independent, it is linearly independent.\label{I:independence:AffLinDep}
    \item If $\exper\in\exspace$ is Bayes-plausible and $\supp\exper$ is linearly independent, then $\supp\exper'\neq\supp\exper$ for any Bayes-plausible $\exper'\in\exspace, \exper'\neq\exper$.\label{I:independence:uniqueindependent}
\end{enumerate}
\end{cor}

\begin{proof}
Follows identically to Lemma \ref{L:independenceT}.
\end{proof}

\begin{proofof}{\bf Corollary \ref{C:SolutionT} (Characterization of Optimal Data Publication Mechanisms)}
From Proposition \ref{P:DesignProblemT},  the maximized value of the designer's problem (\ref{E:OriginalProblem}) is $\max_{\texper\in\Delta(\tconstraint(\epsilon,\tprior))}\{E_\texper \tvf(\tbelief)\text{ s.t. } E_\texper\tbelief=\tprior\}$.  
Since $\tvf(\tbelief)=\tvf_{\tconstraint(\epsilon,\tprior)}(\tbelief)$ on $\tconstraint(\epsilon,\tprior)$, 
this is equal to
$\max_{\texper\in\Delta(\tconstraint(\epsilon,\tprior))}\{E_\texper \tvf_{\tconstraint(\epsilon,\tprior)}(\tbelief)\text{ s.t. }E_\texper\tbelief=\tprior\}.$
By Lemma \ref{L:tconstraint} (\ref{I:tconstraintintprior}), $\tprior\in\ri(\tconstraint(\epsilon,\tprior))$. 
Since $\tconstraint(\epsilon,\tprior)\subseteq\Delta(\tspace)$, $\tconstraint(\epsilon,\tprior)$ is bounded in $\mathbb{R}^{(T+1)^N}$; then by Lemma \ref{L:tconstraint} (\ref{I:tconstraintconvex}), it is compact and convex. 
(\ref{I:SolutionT_Max}) then follows from  Proposition 3 in the Online Appendix to 
\cite{KamenicaBayesianPersuasion2011}.

For (\ref{I:SolutionT_ArgMax}), first 
choose $\texper\in
\arg
\max_{\texper\in\Delta(\tconstraint(\epsilon,\tprior))}\{E_\texper \tvf(\tbelief)\text{ s.t. }E_\texper\tbelief=\tprior\}
=\\
\arg\max_{\texper\in\Delta(\tconstraint(\epsilon,\tprior))}\{E_\texper \tvf_{\tconstraint(\epsilon,\tprior)}(\tbelief)\text{ s.t. }E_\texper\tbelief=\tprior\}$; such a $\texper$ exists by (\ref{I:SolutionT_Max}).
If $\supp\texper\nsubseteq\ext(\tconstraint(\epsilon,\tprior))$, then by Lemma \ref{L:extremalT} (\ref{I:BoundAttainedT}) we can choose a $\hat{\texper}$ with $\supp\hat{\texper}\subseteq\ext(\tconstraint(\epsilon,\tprior))$ which is  Blackwell-more informative than $\texper$.
 If $\supp\texper\subseteq\ext(\tconstraint(\epsilon,\tprior))$, choose $\hat{\texper}=\texper$. Then by Lemma \ref{L:tconstraint}\eqref{I:tconstraintextreme}, $\supp\hat{\texper}$ is finite. Since $\tvf$ is convex, $E_\texper\tvf_{\tconstraint(\epsilon,\tprior)}(\tbelief)=E_\texper\tvf(\tbelief)\leq E_{\hat{\texper}}\tvf(\tbelief)=E_{\hat{\texper}}\tvf_{\tconstraint(\epsilon,\tprior)}(\tbelief)$, and $\hat{\texper}\in\arg\max_{\texper\in\Delta(\tconstraint(\epsilon,\tprior))}\{E_\texper \tvf_{\tconstraint(\epsilon,\tprior)}(\tbelief)\text{ s.t. }E_\texper\tbelief=\tprior\}$.

Now by Lemma \ref{L:independenceT} \eqref{I:independenceT:supportindependence}, there is a Bayes-plausible $\texper^*$ such that $\supp\texper^*$ is an affinely independent set and $\supp\texper^*\subseteq\supp\hat{\texper}\subseteq\ext(\tconstraint(\epsilon,\tprior))$. 
Moreover, by Lemma 3 in \cite{yoder_JMP}, $\tccvf_{\tconstraint(\epsilon,\tprior)}(\tbelief)=\tvf_{\tconstraint(\epsilon,\tprior)}(\tbelief)=\tvf(\tbelief)$ for each $\tbelief\in\supp\hat{\texper}$, and in addition, $\tccvf_{\tconstraint(\epsilon,\tprior)}$ is affine on $\conv(\supp\hat{\texper})$: there exists $x\in\mathbb{R}^{(T+1)^N}$ such that for each $\tbelief\in\conv(\supp\hat{\texper})$, $\tccvf_{\tconstraint(\epsilon,\tprior)}(\tbelief)=\tccvf_{\tconstraint(\epsilon,\tprior)}(\tprior)+x\cdot (\tbelief-\tprior)$. 
Then since $\supp\texper^*\subseteq\supp\hat{\texper}$, we have $E_{\texper^*}\tvf(\tbelief)=E_{\texper^*}\tccvf_{\tconstraint(\epsilon,\tprior)}(\tbelief)=\tccvf_{\tconstraint(\epsilon,\tprior)}(\tprior)+E_{\texper^*}[x\cdot(\tbelief-\tprior)]=\tccvf_{\tconstraint(\epsilon,\tprior)}(\tprior)$, and so\\ $\texper^*\in\arg\max_{\texper\in\Delta(\tconstraint(\epsilon,\tprior))}\{E_\texper \tvf(\tbelief)\text{ s.t. }E_\texper\tbelief=\tprior\}$.

Since $\texper^*$ is Bayes-plausible, has finite support, and satisfies $\supp\texper^*\subseteq\tconstraint(\epsilon,\tprior)$, Lemma \ref{L:DP} implies that it is induced by an $\epsilon$-DP mechanism. By Proposition \ref{P:DesignProblemT}, this mechanism solves the designer's problem \eqref{E:OriginalProblem}.
Since it induces $\texper^*$, (\ref{I:SolutionT_Independence}) follows from affine independence of $\supp\texper^*$ and Lemma \ref{L:independenceT} \eqref{I:independenceT:AffLinDep}; \eqref{I:SolutionT_UniqueSupport} then follows from Lemma \ref{L:independenceT} \eqref{I:independenceT:uniqueindependent}.
And since $\supp\texper^*\subseteq\ext(\tconstraint(\epsilon,\tprior))$, (\ref{I:SolutionT_BoundAttained}) follows from Lemma \ref{L:extremalT} (\ref{I:VerticesT}).
\hfill $\square $\smallskip
\end{proofof}

\subsubsection*{Oblivious vs. General Publication Mechanisms}

\begin{lem}\label{L:OblivProj}
The following are equivalent:
\begin{enumerate}[i.]
    \item If the distribution $\exper\in\exspace$ of posterior beliefs about the population statistic can be induced by an $\epsilon$-DP mechanism, it can be induced by an $\epsilon$-DP oblivious mechanism.\label{I:OblivProj:OblivEquiv}
    \item $\constraint(\epsilon,\prior)=\proj\tconstraint(\epsilon,\tprior)$.\label{I:OblivProj:ProjEquiv}
\end{enumerate}
\end{lem}
\begin{proof}
((\ref{I:OblivProj:OblivEquiv})$\Rightarrow$(\ref{I:OblivProj:ProjEquiv})) We prove the contrapositive.
Suppose that $\constraint(\epsilon,\prior)\neq\proj\tconstraint(\epsilon,\tprior)$. 
Since oblivious $\epsilon$-DP mechanisms are a subset of all $\epsilon$-DP mechanisms, and each belief $\tbelief\in\Delta(\tspace)$ induces the belief about the database $\proj\tbelief\in\Delta(\Omega)$, it follows from Proposition \ref{P:DPObliv} and Lemma \ref{L:DP} that $\constraint(\epsilon,\prior)\subseteq\proj\tconstraint(\epsilon,\tprior)$. So there must exist $\belief\in \proj\tconstraint(\epsilon,\tprior)$ with $\belief\notin \constraint(\epsilon,\prior)$; hence, there must exist $\tbelief\in \tconstraint(\epsilon,\tprior)$ with $\proj\tbelief\notin \constraint(\epsilon,\prior)$.
By Lemma \ref{L:tconstraint} (\ref{I:tconstraintintprior}), there exists $\delta>0$ such that $\tbelief'=\tprior-\delta(\tbelief-\tprior)\in\tconstraint(\epsilon,\tprior)$. Then the distribution $\texper$ with \begin{align*}
    \texper(\tbelief)&=\frac{\delta}{1+\delta}, & \texper(\tbelief')&=\frac{1}{1+\delta}
\end{align*}
is Bayes-plausible with $\supp\texper\subset\tconstraint(\epsilon,\tprior)$, and so by Lemma \ref{L:DP} can be induced by an $\epsilon$-DP mechanism, which thus induces the distribution of posteriors about the population statistic $\exper=\texper\circ\proj^{-1}\in\exspace$. Then $\proj\tbelief\in\supp\exper$, but $\proj\tbelief\notin\constraint(\epsilon,\prior)$; it follows from Proposition \ref{P:DPObliv} that $\exper$ cannot be induced with an oblivious $\epsilon$-DP mechanism.

((\ref{I:OblivProj:ProjEquiv})$\Rightarrow$(\ref{I:OblivProj:OblivEquiv}))
Suppose that $\constraint(\epsilon,\prior)=\proj\tconstraint(\epsilon,\tprior)$, 
and that $\exper\in\exspace$ is the distribution of posterior beliefs about the population statistic induced by the $\epsilon$-DP mechanism $(S,\mech)$.
Then $\exper$ is Bayes-plausible.
Then 
$\supp \exper=\proj\supp\texper$, where $\texper$ is the distribution of posterior beliefs about $\theta$ induced by $(S,\mech)$. Since $(S,\mech)$ is $\epsilon$-DP, by Lemma \ref{L:DP}, $\supp \texper\subset\tconstraint(\epsilon,\tprior)$. Then $\supp \exper=\proj\supp\texper\subset\proj \tconstraint(\epsilon,\tprior)=\constraint(\epsilon,\prior)$; since $\exper$ is Bayes-plausible, it follows from Proposition \ref{P:DPObliv} that $\exper$ can be induced by an oblivious $\epsilon$-DP mechanism.
\end{proof}

\begin{proofof}{\bf Theorem \ref{T:OblivMag} (Oblivious Mechanisms and Magnitude Data)}
    Let $d(\theta,\theta')=|\{i|\theta_i\neq\theta_i'\}|$ be the Hamming distance on $\Theta$.
    For each $\hat{\theta}\in\tspace$, let
    \[\tbelief_{\hat{\theta}}(\theta)=\frac{\tprior(\theta)e^{-\psi(\theta)\epsilon}}{\sum_{t\in\tspace}\tprior(t)e^{-\psi(t)\epsilon}},\text{ where }\psi(\theta)=\left\{\begin{array}{rl}
        \min\{\min\{N,T\},d(\theta,\hat{\theta})\},& \omega_\theta\neq\omega_{\hat{\theta}};\\
        \min\{\min\{N,T\}-1,d(\theta,\hat{\theta})\},& \omega_\theta=\omega_{\hat{\theta}}.
    \end{array}\right.\]
    For any distinct $\theta,\theta'$ with $\theta_{-i}=\theta'_{-i}$ for some $i$, we 
    have $|d(\theta,{\hat{\theta}})-d(\theta',{\hat{\theta}})|\leq1$: changing a single entry of $\theta$ can change the Hamming distance between it and $\hat{\theta}$ by at most $1$. 
    Consequently, since $|\min\{a,b\}-\min\{x,y\}|\leq\max\{|a-x|,|b-y|\}$ for each $a,b,x,y\in\mathbb{R}$, it follows that
    $|\psi(\theta)-\psi(\theta')|\leq\max\{|d(\theta,\hat{\theta})-d(\theta',\hat{\theta})|,1\}\leq1.$
    Hence,
    \[ \left|\log\left(\frac{\tbelief_{\hat{\theta}}(\theta)}{\tbelief_{\hat{\theta}}(\theta')}\right)-\log\left(\frac{\tprior(\theta)}{\tprior(\theta')}\right)\right|=\left|\log\left(\frac{e^{-\psi(\theta)\epsilon}}{e^{-\psi(\theta')\epsilon}}\right)\right|=|\psi(\theta)-\psi(\theta')|\epsilon\leq \epsilon,
    \]
    and so $\tbelief_{\hat{\theta}}\in\tconstraint(\epsilon,\tprior)$. 

    Now let $\hat{\theta}$ be a database with $\hat{\theta}_i=1$ for $i\leq\min\{\hitype,N\}$ and $\hat{\theta}_i=0$ for $i>\min\{\hitype,N\}$. Then $\omega_{\hat{\theta}}=\min\{T,N\}$.
    Then (since $\prior(0)=\tprior(\textbf{0})$) we have
    \begin{align*}
        \left|\log\left(\frac{\proj\tbelief_{\hat{\theta}}(\omega_{\hat{\theta}})}{\proj\tbelief_{\hat{\theta}}(0)}\right)-\log\left(\frac{\prior(\omega_{\hat{\theta}})}{\prior(0)}\right)\right|&=\left|\log\left(\frac{\sum_{\theta:\omega_\theta=\omega_{\hat{\theta}}}\tprior(\theta)e^{-\psi(\theta)\epsilon}}{\prior(0)e^{-\min\{\hitype,N\}\epsilon}}\right)-\log\left(\frac{\prior(\omega_{\hat{\theta}})}{\prior(0)}\right)\right|\\
        &=\left|\log\left(\frac{\sum_{\theta:\omega_\theta=\omega_{\hat{\theta}}}\tprior(\theta)e^{-\psi(\theta)\epsilon}}{\prior(\omega_{\hat{\theta}})e^{-\min\{\hitype,N\}\epsilon}}\right)\right|.
    \end{align*}
    Now note that for any $\theta$ with $\omega_\theta=\omega_{\hat{\theta}}=\min\{\hitype,N\}$, $\psi(\theta)\leq\min\{\hitype,N\}-1$, and since $T>1$, $\psi(\hat{\theta})=0<\min\{\hitype,N\}-1$. Then
    \[\log\left(\frac{\sum_{\theta:\omega_\theta=\omega_{\hat{\theta}}}\tprior(\theta)e^{-\psi(\theta)\epsilon}}{\prior(\omega_{\hat{\theta}})e^{-\min\{\hitype,N\}\epsilon}}\right)>\log\left(\frac{e^\epsilon\sum_{\theta:\omega_\theta=\min\{\hitype,N\}}\tprior(\theta)e^{-\min\{\hitype,N\}\epsilon}}{\prior(\min\{\hitype,N\})e^{-\min\{\hitype,N\}\epsilon}}\right)=\epsilon.\]
    Then $\proj\tbelief_{\hat{\theta}}$ does not satisfy \eqref{E:DP_posterior} for $\omega=t=\min\{N,T\}\leq T$, and so by definition, $\proj\tbelief_{\hat{\theta}}\notin\constraint(\epsilon,\prior)$. Then $\constraint(\epsilon,\prior)\neq\proj\tconstraint(\epsilon,\tprior)$; the rest of the claim follows from Lemma \ref{L:OblivProj}.\hfill$\square$\smallskip
\end{proofof}

\begin{proofof}{\bf Corollary \ref{C:OblivSuboptimal} (Suboptimality of Oblivious Mechanisms)}
By Theorem \ref{T:OblivMag}, choose $\hat{\belief}\in\proj\tconstraint(\epsilon,\tprior)\setminus\constraint(\epsilon,\prior)$, and choose $\hat{\tbelief}\in\tconstraint(\epsilon,\tprior)$ such that $\proj\hat{\tbelief}=\hat{\belief}$. 
By Corollary \ref{C:constraint}, $\constraint(\epsilon,\prior)$ is closed and convex. 
Then by the separating hyperplane theorem, there exists $z\in\mathbb{R}^{NT+1}$ such that $z\cdot \hat{\belief}>0>z\cdot \belief$ for each $\belief\in\constraint(\epsilon,\prior)$. Define $u:\{0,1\}\times\Omega\to\mathbb{R}$ by 
$u(1,\omega)=z_\omega$ and $u(0,\omega)=0$.
Then $\vf(\belief)=0$ for every $\belief\in\constraint(\epsilon,\prior)$, while $\tvf(\hat{\tbelief})=\vf(\hat{\belief})>0$ by \eqref{E:vfproj}. Moreover, $\tvf(\tbelief)\geq 0$ for all $\tbelief\in\Delta(\tspace)$.

By Lemma \ref{L:tconstraint} (\ref{I:tconstraintintprior}), there exists $\delta>0$ such that $\tbelief'=\tprior-\delta(\hat{\tbelief}-\tprior)\in\tconstraint(\epsilon,\tprior)$. Let $\texper$ place probability $\delta/(1+\delta)$ on $\hat{\tbelief}$ and probability $1/(1+\delta)$ on $\tbelief'$. Then $E_{\texper}\tbelief=\tprior$ and $\supp\texper\subseteq\tconstraint(\epsilon,\tprior)$, so Lemma \ref{L:DP} implies that $\texper$ can be induced by an $\epsilon$-DP mechanism. Moreover,
\[
\max_{\texper\in\Delta(\tconstraint(\epsilon,\tprior))}\{E_\texper \tvf(\tbelief)\text{ s.t. }E_\texper\tbelief=\tprior\}\geq E_{\texper}\tvf(\tbelief)=\frac{\delta}{1+\delta}\tvf(\hat{\tbelief})+\frac{1}{1+\delta}\tvf(\tbelief')>0.
\]
Suppose that $(S,\mech)$ is an $\epsilon$-DP oblivious mechanism; let $\texper'\in\Delta(\tconstraint(\epsilon,\tprior))$ be the distribution of posteriors about the database that it induces. By Proposition \ref{P:DPObliv}, $P\tbelief\in \constraint(\epsilon,\prior)$ for each $\tbelief\in\supp\texper'$. Then $E_{\texper'} \tvf(\tbelief)=0$, and so $\texper'$ does not solve \eqref{E:DesignProblemT}. Then by Proposition \ref{P:DesignProblemT}, $(S,\mech)$ does not solve the designer's problem (\ref{E:OriginalProblem}).\hfill$\square$\smallskip
\end{proofof}

\begin{proofof}{\bf Theorem \ref{T:OblivEquiv} (Oblivious Mechanisms and Binary Data)}
Fix $\tbelief\in\tconstraint(\epsilon,\tprior)$, let $G$ be the set of permutations of $\{1,\ldots,N\}$, and for each $g\in G$, let $g\theta\equiv(\theta_{g(1)},\ldots,\theta_{g(N)})$ and $\tbelief_g(\theta)\equiv\tbelief(g\theta)$. Since respondents are anonymous, $\tprior(g\theta)=\tprior(\theta)$. Moreover, $g\theta$ and $g\theta'$ differ in one entry whenever $\theta$ and $\theta'$ do. Hence,
\[
\left|\log\left(\frac{\tbelief_g(\theta)}{\tbelief_g(\theta')}\right)-\log\left(\frac{\tprior(\theta)}{\tprior(\theta')}\right)\right|
=\left|\log\left(\frac{\tbelief(g\theta)}{\tbelief(g\theta')}\right)-\log\left(\frac{\tprior(g\theta)}{\tprior(g\theta')}\right)\right|\leq\epsilon
\]
for every adjacent $\theta,\theta'$, and so $\tbelief_g\in\tconstraint(\epsilon,\tprior)$. By Lemma \ref{L:tconstraint} (\ref{I:tconstraintconvex}), we have $\bar{\tbelief}\equiv |G|^{-1}\sum_{g\in G}\tbelief_g\in\tconstraint(\epsilon,\tprior)$ as well.

Observe that $\omega_{g\theta}=\omega_\theta$ for each $g\in G$. Since each $g\in G$ is bijective, it follows that
\[
\proj\bar{\tbelief}(\omega)=\sum_{\theta:\omega_\theta=\omega}\frac{1}{|G|}\sum_{g\in G}\tbelief(g\theta)=\frac{1}{|G|}\sum_{g\in G}\sum_{\theta:\omega_\theta=\omega}\tbelief(g\theta)=\frac{1}{|G|}\sum_{g\in G}\sum_{\theta:\omega_\theta=\omega}\tbelief(\theta)=\proj\tbelief(\omega),
\]
for each $\omega\in\Omega$, $g\in G$, and $\theta\in\tspace$. With binary data, $\omega_\theta=\omega_{\theta'}$ if and only if $\theta=g\theta'$ for some $g\in G$.
Thus $\bar{\tbelief}$ (by construction) and $\tprior$ (by hypothesis) are constant on each set $\{\theta:\omega_\theta=\omega\}$, so $\bar{\tbelief}(\theta)/\tprior(\theta)=\proj\tbelief(\omega_\theta)/\prior(\omega_\theta)$ for every $\theta\in\tspace$.
For each $\omega\in\Omega\setminus\{0\}$, choose adjacent $\theta,\theta'$ with $\omega_\theta=\omega$ and $\omega_{\theta'}=\omega-1$. Since $\bar{\tbelief}\in\tconstraint(\epsilon,\tprior)$,
\[
\left|\log\left(\frac{\proj\tbelief(\omega)}{\proj\tbelief(\omega-1)}\right)-\log\left(\frac{\prior(\omega)}{\prior(\omega-1)}\right)\right|
=\left|\log\left(\frac{\bar{\tbelief}(\theta)}{\bar{\tbelief}(\theta')}\right)-\log\left(\frac{\tprior(\theta)}{\tprior(\theta')}\right)\right|\leq\epsilon.
\]
Hence, $\proj\tbelief\in\constraint(\epsilon,\prior)$, and so $\proj\tconstraint(\epsilon,\tprior)\subseteq\constraint(\epsilon,\prior)$. 
Since oblivious $\epsilon$-DP mechanisms are a subset of all $\epsilon$-DP mechanisms, we also have $\constraint(\epsilon,\prior)\subseteq\proj\tconstraint(\epsilon,\tprior)$ by Proposition \ref{P:DPObliv} and Lemma \ref{L:DP}. Thus, $\proj\tconstraint(\epsilon,\tprior)=\constraint(\epsilon,\prior)$; the statement follows by Lemma \ref{L:OblivProj}.
\end{proofof}

\begin{proofof}{\bf Corollary \ref{C:DesignProblem} (Differentially Private Data Publication as Information Design)}
Follows immediately from \eqref{E:vfproj}, Proposition \ref{P:DPObliv}, and Theorem \ref{T:OblivEquiv}.\hfill $\square $\smallskip
\end{proofof}

\begin{proofof}{\bf Proposition \ref{P:OblivEquiv2} (Oblivious Mechanisms with Two Respondents)}
Suppose $\tbelief\in\tconstraint(\epsilon,\tprior)$. Then we have $\proj\tbelief(2)/\prior(2)=\tbelief((1,1))/\tprior((1,1))$, and hence
\begin{align*}
    e^{-\epsilon}\tbelief((0,1))\leq\tbelief((1,1))\frac{\tprior((0,1))}{\tprior((1,1))}&\leq e^{\epsilon}\tbelief((0,1)),\qquad
    e^{-\epsilon}\tbelief((1,0))\leq\tbelief((1,1))\frac{\tprior((1,0))}{\tprior((1,1))}\leq e^{\epsilon}\tbelief((1,0));\\
    \Rightarrow e^{-\epsilon}(\tbelief((1,0))+\tbelief((0,1)))&\leq \tbelief((1,1))\frac{\tprior((1,0))+\tprior((0,1))}{\tprior((1,1))}\leq e^{\epsilon}(\tbelief((1,0))+\tbelief((0,1)))\\
    \Rightarrow e^{-\epsilon}\proj\tbelief(1)&\leq \proj\tbelief(2)\frac{\prior(1)}{\prior(2)}\leq e^{\epsilon}\proj\tbelief(1).
\end{align*}
Symmetrically, $\frac{\proj\tbelief(0)/\proj\tbelief(1)}{\prior(0)/\prior(1)}\in[e^{-\epsilon},e^\epsilon]$. It follows that $\proj\tbelief\in\constraint(\epsilon,\prior)$. 
Then $\proj\tconstraint(\epsilon,\tprior)\subseteq \constraint(\epsilon,\prior)$. Since oblivious $\epsilon$-DP mechanisms are a subset of all $\epsilon$-DP mechanisms, we also have $\constraint(\epsilon,\prior)\subseteq\proj\tconstraint(\epsilon,\tprior)$ by Proposition \ref{P:DPObliv} and Lemma \ref{L:DP}. Thus, $\proj\tconstraint(\epsilon,\tprior)=\constraint(\epsilon,\prior)$; the statement follows by Lemma \ref{L:OblivProj}.
\hfill $\square $\smallskip
\end{proofof}

\begin{proofof}{\bf Corollary \ref{C:Solution} (Characterization of Optimal Oblivious Mechanisms)}
Follows identically to Corollary \ref{C:SolutionT}, relying on Corollary \ref{C:DesignProblem}, Proposition \ref{P:DPObliv}, and Corollaries \ref{C:constraint}, \ref{C:extremal}, and \ref{C:independence} instead of Proposition \ref{P:DesignProblemT}, Lemma \ref{L:DP}, and Lemmas \ref{L:tconstraint}, \ref{L:extremalT}, and \ref{L:independenceT}, respectively.
\end{proofof}

\subsubsection*{Geometric Mechanisms and UPRR Dominance}

\begin{proofof}{\bf Lemma \ref{L:Unimodal} (Posteriors Produced by the Geometric Mechanism)}
Suppose that $\belief$ is the decision maker's posterior belief after observing the realization $s$ of the $\epsilon$-geometric data publication mechanism. 
By Bayes' rule, for each $\omega\in\Omega$, \[\belief(\omega)=\frac{\gmech(s|\omega)\prior(\omega)}{\sum_{\omega'\in\Omega}\gmech(s|\omega')\prior(\omega')}.\]
Then for each $\omega\in\Omega\setminus\{0\}$, we have
\begin{align*}\log\left(\frac{\belief(\omega)/\belief(\omega-1)}{\prior(\omega)/ \prior(\omega-1)}\right)
&=\log\left(\frac{\gmech(s|\omega)}{\gmech(s|\omega-1)}\right)=-\epsilon|s-\omega|+\epsilon|s-(\omega-1)|=\left\{\begin{array}{rl}
\epsilon, &s\geq\omega;\\
-\epsilon, &s<\omega.
\end{array}\right.\end{align*}
The claim follows by letting $x=\min\{\max\{0,s\},N\}$.
\hfill $\square $\smallskip
\end{proofof}

\begin{proofof}{\bf Lemma \ref{L:Transitivity} (Transitivity of UPRR Dominance)}
Fix $\belief\in\supp\exper$, and let $\omega^*(\belief)$ be 
its peak used to show $\exper\succeq_{UPRR}\exper'$.
We first show that there exist $\belief'\in\supp\exper'$ and $\omega'\leq \omega^*(\belief)$ such that $\belief'(\omega)/\belief''(\omega)$ is nonincreasing on $[\omega',\infty)$ for every $\belief''\in\supp\exper''$.

For each $\belief'\in\supp\exper'$, let $w(\belief')$ be a peak used to show $\exper'\succeq_{UPRR}\exper''$. Then for each $\belief''\in\supp\exper''$, $\belief'(\omega)/\belief''(\omega)$ is nondecreasing on $(-\infty,w(\belief')]$ and nonincreasing on $[w(\belief'),\infty)$.
If $w(\belief')\leq \omega^*(\belief)$ for any $\belief'\in\supp\exper'$, then the claim follows by taking $\omega'=w(\belief')$.

Suppose instead that $w(\belief')>\omega^*(\belief)$ for each  $\belief'\in\supp\exper'$. Then for each $\belief'\in\supp\exper'$, $\frac{\belief'(w(\belief'))}{\belief''(w(\belief'))}\geq\frac{\belief'(\omega^*(\belief))}{\belief''(\omega^*(\belief))}$ for every $\belief''\in\supp\exper''$. If there is some $\belief'\in\supp\exper'$ for which this inequality binds at every  $\belief''\in\supp\exper''$, then the claim follows for $\belief'$ by taking $\omega'=\omega^*(\belief)$. 

Suppose toward a contradiction that there is not. Then 
since $\Omega$ is finite, choose $\hat{\belief}'\in\arg\min_{\belief'\in\supp\exper'}w(\belief')$.
Then, for every $\belief'\in\supp\exper'$ and $\belief''\in\supp\exper''$,
\[
\frac{\belief'(w(\hat{\belief}'))}{\belief''(w(\hat{\belief}'))}\geq\frac{\belief'(\omega^*(\belief))}{\belief''(\omega^*(\belief))},
\]
and the inequality is strict at $\belief'=\hat{\belief}'$ for some $\bar{\belief}''\in\supp\exper''$.
Bayes plausibility therefore gives, for every $\belief''\in\supp\exper''$,
\begin{align*}
\frac{\prior(w(\hat{\belief}'))}{\belief''(w(\hat{\belief}'))}
&=\frac{\int_{\supp\exper'}\belief'(w(\hat{\belief}'))\,d\exper'(\belief')}{\belief''(w(\hat{\belief}'))}
\geq\frac{\int_{\supp\exper'}\belief'(\omega^*(\belief))\,d\exper'(\belief')}{\belief''(\omega^*(\belief))}
=\frac{\prior(\omega^*(\belief))}{\belief''(\omega^*(\belief))}.
\end{align*}
For $\belief''=\bar{\belief}''$, the inequality is strict: the difference between the integrands is continuous in $\belief'$ and positive at $\hat{\belief}'$, so it is positive on a neighborhood that has positive $\exper'$-measure because $\hat{\belief}'\in\supp\exper'$. Taking reciprocals and applying the same continuity-and-support argument under $\exper''$ yields
\[
1=\frac{\int_{\supp\exper''}\belief''(w(\hat{\belief}'))\,d\exper''(\belief'')}{\prior(w(\hat{\belief}'))}
<
\frac{\int_{\supp\exper''}\belief''(\omega^*(\belief))\,d\exper''(\belief'')}{\prior(\omega^*(\belief))}=1,
\]
a contradiction.

So there exists $\belief'\in\supp\exper'$ and $\omega'\leq\omega^*(\belief)$ such that $\belief'(\omega)/\belief''(\omega)$ is nonincreasing on $[\omega',\infty)$ for every $\belief''\in\supp\exper''$.
Then for every $\hat{\omega}\geq\omega\geq \omega^*(\belief)$ and $\belief''\in\supp\exper''$,
\[
\frac{\belief(\hat{\omega})}{\belief''(\hat{\omega})}
=\frac{\belief(\hat{\omega})}{\belief'(\hat{\omega})}
\frac{\belief'(\hat{\omega})}{\belief''(\hat{\omega})}
\leq
\frac{\belief(\omega)}{\belief'(\omega)}
\frac{\belief'(\omega)}{\belief''(\omega)}
=\frac{\belief(\omega)}{\belief''(\omega)}.
\]
A symmetric argument shows that
this ratio is nondecreasing below $\omega^*(\belief)$. Since this holds for every $\belief\in\supp\exper$, $\exper\succeq_{UPRR}\exper''$.\hfill$\square$
    
\end{proofof}

\begin{proofof}{\bf Lemma \ref{L:Geometric_Dominates} (Geometric Mechanisms are UPRR-Dominant)} 
Proposition \ref{P:DPObliv} and Corollary \ref{C:constraint} (\ref{I:constraintintsimplex}) imply that every posterior in $\supp\gpost\cup\supp\exper$ has full support. Fix $\belief\in\supp\gpost$ and $\belief'\in\supp\exper$.
Lemma \ref{L:Unimodal} and Proposition \ref{P:DPObliv} imply, for every $\omega\in\Omega\setminus\{0\}$,
\[
\frac{\belief(\omega)/\belief(\omega-1)}{\prior(\omega)/\prior(\omega-1)}=
\begin{cases}e^\epsilon,&\omega\leq\omega^*(\belief),\\ e^{-\epsilon},&\omega>\omega^*(\belief).\end{cases}
\qquad
e^{-\epsilon}\leq\frac{\belief'(\omega)/\belief'(\omega-1)}{\prior(\omega)/\prior(\omega-1)}\leq e^\epsilon.
\]
Consequently,
\[
\frac{\belief(\omega)/\belief'(\omega)}{\belief(\omega-1)/\belief'(\omega-1)}\geq1\quad\text{if }0<\omega\leq \omega^*(\belief),
\qquad
\frac{\belief(\omega)/\belief'(\omega)}{\belief(\omega-1)/\belief'(\omega-1)}\leq1\quad\text{if }\omega>\omega^*(\belief).
\]
Thus, $\belief(\omega)/\belief'(\omega)$ is nondecreasing up to $\omega^*(\belief)$ and nonincreasing thereafter. Since this holds for every $\belief\in\supp\gpost$ and $\belief'\in\supp\exper$, $\gpost\succeq_{UPRR}\exper$.\hfill $\square $\smallskip
\end{proofof}

\begin{proofof}{\bf Theorem \ref{T:ProbDominance} (UPRR Dominance Implies Dominance in Monotone Decision Problems)}
For each $\belief\in\supp\exper$, fix a peak $\omega^*(\belief)$ that verifies $\exper\succeq_{UPRR}\exper'$. Bayes plausibility and the full-support assumption imply that $\prior$ has full support.
For each $\belief\in\supp\exper'$, choose
$a^*(\belief)\in\arg\max_{a\in A}E_\belief[u(a,\omega)]$, which is nonempty by continuity and compactness.

For any finite-support $\eta\in\Delta(\Delta(\Omega))$ and any $t:\supp\eta\to\mathbb{R}$, 
define the cdf and quantile of $t$ under $\eta$ as \begin{align*}
\cdf_{\eta,t}(r)&:=\eta(\{\belief\mid t(\belief)\leq r\}),&\quant_{\eta,t}(z)&:=
\min\{r\in t(\supp\eta)\mid \cdf_{\eta,t}(r)\geq z\}.
\end{align*}
Then for every $r\in t(\supp\eta)$, \begin{align}
    \int_{\{z\mid\quant_{\eta,t}(z)=r\}}dz&=\cdf_{\eta,t}(r)-\lim_{x\uparrow r}\cdf_{\eta,t}(x)=\eta\{\belief\mid t(\belief)=r\}.\label{E:quantmeasure}
\end{align}

Define two Blackwell experiments with signal space $[0,1]$ by the conditional densities
\begin{align*}
f(z\mid\omega)
&=
E_\exper\left[\frac{\belief(\omega)}{\prior(\omega)}\mid \omega^*(\belief)=\quant_{\exper,\omega^*}(z)\right]
=\sum_{\substack{\belief\in\supp\exper\\
\omega^*(\belief)=\quant_{\exper,\omega^*}(z)}}
\frac{\exper(\belief)\belief(\omega)}
{\prior(\omega)\exper(\{\tilde{\belief}:\omega^*(\tilde{\belief})=\quant_{\exper,\omega^*}(z)\})},\\
g(z\mid\omega)
&=
E_{\exper'}\left[\frac{\belief(\omega)}{\prior(\omega)}\mid a^*(\belief)=\quant_{\exper',a^*}(z)\right]
=\sum_{\substack{\belief\in\supp\exper'\\
a^*(\belief)=\quant_{\exper',a^*}(z)}}
\frac{\exper'(\belief)\belief(\omega)}
{\prior(\omega)\exper'(\{\tilde{\belief}:a^*(\tilde{\belief})=\quant_{\exper',a^*}(z)\})}.
\end{align*}

Observe that $f(\cdot|\omega)$ is constant on $\{z\mid \quant_{\exper,\omega^*}(z)=r\}$, and $g(\cdot|\omega)$ is constant on $\{z\mid \quant_{\exper',a^*}(z)=r\}$, for each $r$ and each $\omega$.
Then for every $\omega\in\Omega$, \eqref{E:quantmeasure} and Bayes plausibility give
\[
\int_0^1 f(z\mid\omega)\,dz
=\sum_{\belief\in\supp\exper}\frac{\exper(\belief)\belief(\omega)}{\prior(\omega)}=1,
\qquad
\int_0^1 g(z\mid\omega)\,dz
=\sum_{\belief\in\supp\exper'}\frac{\exper'(\belief)\belief(\omega)}{\prior(\omega)}=1.
\]
Let $F(x\mid\omega)=\int_0^x f(z\mid\omega)\,dz$ and $G(x\mid\omega)=\int_0^x g(z\mid\omega)\,dz$. Since each $\belief\in\supp\exper\cup\supp\exper'$ has full support, $f(z\mid\omega)>0$ and $g(z\mid\omega)>0$ for every $\omega\in\Omega$ and almost every $z\in[0,1]$; it follows that $F(\cdot\mid\omega)$ and $G(\cdot\mid\omega)$ are continuous and strictly increasing and have support $[0,1]$.

For almost every $z,v\in[0,1]$ and every $\omega,\omega'\in\Omega$,
\begin{align*}
&f(z\mid\omega')g(v\mid\omega)-f(z\mid\omega)g(v\mid\omega')\\
&=\frac{E_{\belief\sim\exper,\belief'\sim\exper'}\left[\belief(\omega')\belief'(\omega)-\belief(\omega)\belief'(\omega')\mid \omega^*(\belief)=\quant_{\exper,\omega^*}(z),a^*(\belief')=\quant_{\exper',a^*}(v)\right]}{\prior(\omega)\prior(\omega')}.
\end{align*}
Since $\exper\succeq_{UPRR}\exper'$, for every $\belief\in\supp\exper$ and $\belief'\in\supp\exper'$, we have
\[
\belief(\omega')\belief'(\omega)-\belief(\omega)\belief'(\omega')
\begin{cases}
\geq0,&\omega\leq\omega'\leq\omega^*(\belief),\\
\leq0,&\omega^*(\belief)\leq\omega\leq\omega'.
\end{cases}
\]
Hence, for almost every $z,v\in[0,1]$,
\begin{align}
    f(z\mid\omega')g(v\mid\omega)\geq f(z\mid\omega)g(v\mid\omega')
    &\text{ if }\omega\leq\omega'\leq\quant_{\exper,\omega^*}(z),\label{E:uprr_below}\\
    f(z\mid\omega')g(v\mid\omega)\leq f(z\mid\omega)g(v\mid\omega')
    &\text{ if }\quant_{\exper,\omega^*}(z)\leq\omega\leq\omega'\label{E:uprr_above}.
\end{align}

Fix $x,y\in(0,1)$. Since $\quant_{\exper,\omega^*}$ is nondecreasing, 
$\omega\leq\omega'\leq\quant_{\exper,\omega^*}(x)$ implies $\omega\leq\omega'\leq\quant_{\exper,\omega^*}(z)$ for all $z\geq x$; similarly, $\quant_{\exper,\omega^*}(x)\leq\omega\leq\omega'$ implies $\quant_{\exper,\omega^*}(z)\leq\omega\leq\omega'$ for all $z\leq x$.
Integrating \eqref{E:uprr_below} on $[x,1]\times[y,1]$ and \eqref{E:uprr_above} on $[0,x]\times[0,y]$ thus yields
\begin{align}
[1-F(x\mid\omega')][1-G(y\mid\omega)]
\geq[1-F(x\mid\omega)][1-G(y\mid\omega')]
&\text{ if }\omega\leq\omega'\leq\quant_{\exper,\omega^*}(x),\label{E:uprr_int_below}\\
F(x\mid\omega')G(y\mid\omega)
\leq F(x\mid\omega)G(y\mid\omega')
&\text{ if }\quant_{\exper,\omega^*}(x)\leq\omega\leq\omega'\label{E:uprr_int_above}.
\end{align}

Now suppose toward a contradiction that $F(x\mid\omega)\leq G(y\mid\omega)$ but $F(x\mid\omega')>G(y\mid\omega')$ for some $\omega'\geq\omega$.
If $\omega\leq\omega'\leq\quant_{\exper,\omega^*}(x)$, then
\[
[1-F(x\mid\omega')][1-G(y\mid\omega)]
<[1-G(y\mid\omega')][1-G(y\mid\omega)]
\leq[1-G(y\mid\omega')][1-F(x\mid\omega)],
\]
contradicting \eqref{E:uprr_int_below}. If $\quant_{\exper,\omega^*}(x)\leq\omega\leq\omega'$, then
\[
F(x\mid\omega')G(y\mid\omega)
>G(y\mid\omega')G(y\mid\omega)
\geq G(y\mid\omega')F(x\mid\omega),
\]
contradicting \eqref{E:uprr_int_above}. 
Then if $\quant_{\exper,\omega^*}(x)\leq\omega\leq\omega'$ or $\omega\leq\omega'\leq\quant_{\exper,\omega^*}(x)$, we have $F(x\mid\omega)\leq G(y\mid\omega)\Rightarrow F(x\mid\omega')\leq G(y\mid\omega')$. It follows that if  $\omega<\quant_{\exper,\omega^*}(x)<\omega'$, then
\[
F(x\mid\omega)\leq G(y\mid\omega)
\Longrightarrow F(x\mid\quant_{\exper,\omega^*}(x))\leq G(y\mid\quant_{\exper,\omega^*}(x))
\Longrightarrow F(x\mid\omega')\leq G(y\mid\omega').
\]
Therefore,  for every $x,y\in(0,1)$ and $\omega\leq\omega'$, we have $F(x\mid\omega)\leq G(y\mid\omega)\Rightarrow F(x\mid\omega')\leq G(y\mid\omega')$.
Continuity and the identities $F(0\mid\omega)=G(0\mid\omega)=0$ and $F(1\mid\omega)=G(1\mid\omega)=1$ extend this implication to $x,y\in[0,1]$.

Define $F^{-1}(p\mid\omega)=\inf\{x\in[0,1]:F(x\mid\omega)\geq p\}$. Strict monotonicity and continuity imply
\[
x\leq F^{-1}(G(y\mid\omega)\mid\omega)
\quad\Longleftrightarrow\quad
F(x\mid\omega)\leq G(y\mid\omega).
\]
Taking $x=F^{-1}(G(y\mid\omega)\mid\omega)$ in the preceding implication gives, for every $y\in[0,1]$ and $\omega\leq\omega'$,
\[
F^{-1}(G(y\mid\omega)\mid\omega)
\leq F^{-1}(G(y\mid\omega')\mid\omega').
\]
Thus, $F$ is more accurate than $G$ in the sense of \cite{lehmann1988comparing}.

The rule $y\mapsto\quant_{\exper',a^*}(y)$ is nondecreasing.
Since $\Omega$ is finite and $u(\cdot,\omega)$ is continuous on compact $A$, the IDO family is equicontinuous. Proposition 9 of \cite{quah2009comparative} therefore implies that there is a nondecreasing rule $\alpha:[0,1]\to A$ such that, for every $\omega\in\Omega$,
\[
\int_0^1u(\alpha(x),\omega)\,dF(x\mid\omega)
\geq
\int_0^1u(\quant_{\exper',a^*}(y),\omega)\,dG(y\mid\omega).
\]
For each $\belief\in\supp\exper$,
\[
\max_{a\in A}E_\belief[u(a,\omega)]
\geq \frac{1}{\exper(\{\tilde{\belief}:\omega^*(\tilde{\belief})=\omega^*(\belief)\})}
\int_{\{x:\quant_{\exper,\omega^*}(x)=\omega^*(\belief)\}}
E_\belief[u(\alpha(x),\omega)]\,dx.
\]
Multiplying these inequalities by $\exper(\belief)$, summing over $\supp\exper$, and using \eqref{E:quantmeasure} and the definitions of $f$ and $g$ gives
\begin{align*}
E_\exper\left[\max_{a\in A}E_\belief[u(a,\omega)]\right]
&\geq\sum_{\omega\in\Omega}\prior(\omega)\int_0^1u(\alpha(x),\omega)\,dF(x\mid\omega)\\
&\geq\sum_{\omega\in\Omega}\prior(\omega)\int_0^1u(\quant_{\exper',a^*}(y),\omega)\,dG(y\mid\omega)\\
&=\sum_{\belief\in\supp\exper'}\exper'(\belief)\sum_{\omega\in\Omega}\belief(\omega)u(a^*(\belief),\omega)=E_{\exper'}\left[\max_{a\in A}E_\belief[u(a,\omega)]\right],
\end{align*}
as desired.
\hfill $\square $\smallskip
\end{proofof}

\begin{proofof}{\bf Theorem \ref{T:Geometric} (Optimality of the Geometric Mechanism for Monotone Decision Problems)}
By Corollary \ref{C:Solution} (\ref{I:Solution_Independence}), the designer's problem (\ref{E:OriginalProblem}) is solved by an $\epsilon$-DP oblivious mechanism which induces a distribution of posteriors about the population statistic $\exper^*$ whose support is linearly independent, and hence finite. By Lemma \ref{L:Unimodal}, $\supp\gpost$ is also finite.
By Proposition \ref{P:DPObliv}, every posterior induced by this mechanism or the geometric mechanism has full support. By Lemma \ref{L:Geometric_Dominates}, $\gpost\succeq_{UPRR}\exper^*$. The claim follows from Theorem \ref{T:ProbDominance}.
\hfill $\square $\smallskip
\end{proofof}

\end{document}